%% file: arxiv.tex
\documentclass{article}

\usepackage{microtype}
\usepackage{graphicx}
\usepackage{subfigure}
\usepackage{booktabs} 
\usepackage[colorlinks,citecolor=blue,bookmarks=true,pagebackref=true, urlcolor=blue]{hyperref}
\usepackage[numbers,square,sort]{natbib}
\usepackage{algorithm}
\usepackage{algorithmic}
\usepackage[margin=1in]{geometry}
\usepackage[semibold]{sourceserifpro}

\usepackage{amsmath}
\usepackage{amssymb}
\usepackage{mathtools}
\usepackage{amsthm}

\usepackage[capitalize,noabbrev]{cleveref}

\theoremstyle{plain}
\newtheorem{theorem}{Theorem}[section]

\newtheorem{lemma}[theorem]{Lemma}

\theoremstyle{definition}
\newtheorem{definition}[theorem]{Definition}

\theoremstyle{remark}

\newcommand{\frnorm}[1]{\|#1\|_{\mathsf{F}}}
\newcommand{\opnorm}[1]{\|#1\|_{{2}}}
\newcommand{\onenorm}[1]{\|#1\|_{1}}
\newcommand{\opt}{\textnormal{Opt}}
\newcommand{\R}{\mathbb{R}}

\newcommand{\T}[1]{#1^{\mathsf{T}}}

\newcommand{\poly}{\operatorname{poly}}

\renewcommand{\Pr}{\text{Pr}}
\newcommand{\bS}{\mathbf{S}}
\newcommand{\bU}{\mathbf{U}}
\newcommand{\bA}{\mathbf{A}}
\newcommand{\bM}{\mathbf{M}}
\newcommand{\E}{\mathbb{E}}

\newcommand{\bX}{\mathbf{X}}
\newcommand{\bV}{\mathbf{V}}
\newcommand{\bZ}{\mathbf{Z}}
\newcommand{\nnz}{\textnormal{\texttt{nnz}}}
\newcommand{\Var}{\text{Var}}

\newcommand{\br}{\mathbf{r}}
\newcommand{\bg}{\mathbf{g}}
\newcommand{\bQ}{\mathbf{Q}}
\newcommand{\calS}{\mathcal{S}}

\newcommand{\bG}{\mathbf{G}}

\newcommand{\cost}{\text{cost}}
\renewcommand{\epsilon}{\varepsilon}

\DeclareMathOperator*{\argmin}{arg\,min}
\newcommand{\set}[1]{\{#1\}}

\title{Sketching Algorithms and Lower Bounds for Ridge Regression}
\author{Praneeth Kacham \\ CMU \\ \href{mailto:pkacham@cs.cmu.edu}{\texttt{pkacham@cs.cmu.edu}} \and David P. Woodruff \\ CMU \\ \href{mailto:dwoodruf@cs.cmu.edu}{\texttt{dwoodruf@cs.cmu.edu}}}
\date{}
\begin{document}
\maketitle

\input{main_content}
\bibliographystyle{plainnat}
\bibliography{main} 
\newpage
\appendix
\input{appendix}
\end{document}

%% file: main_content.tex
\begin{abstract}
    We give a sketching-based iterative algorithm that computes a $1+\varepsilon$ approximate solution for the ridge regression problem $\min_x \opnorm{Ax-b}^2 +\lambda\opnorm{x}^2$ where $A \in \R^{n \times d}$ with $d \ge n$. Our algorithm, for a constant number of iterations (requiring a constant number of passes over the input), improves upon earlier work \cite{chowdhury2018iterative} by requiring that the sketching matrix only has a weaker Approximate Matrix Multiplication (AMM) guarantee that depends on $\epsilon$, along with a constant subspace embedding guarantee. The earlier work instead requires that the sketching matrix has a subspace embedding guarantee that depends on $\epsilon$. For example, to produce a $1+\varepsilon$ approximate solution in $1$ iteration, which requires $2$ passes over the input, our algorithm requires the OSNAP embedding to have $m= O(n\sigma^2/\lambda\varepsilon)$ rows with a sparsity parameter $s = O(\log(n))$, whereas the earlier algorithm of \citet{chowdhury2018iterative} with the same number of rows of OSNAP requires a sparsity $s = O(\sqrt{\sigma^2/\lambda\varepsilon} \cdot \log(n))$, where $\sigma = \opnorm{A}$ is the spectral norm of the matrix $A$. We also show that this algorithm can be used to give faster algorithms for kernel ridge regression. Finally, we show that the sketch size required for our algorithm is essentially optimal for a natural framework of algorithms for ridge regression by proving lower bounds on oblivious sketching matrices for AMM. The sketch size lower bounds for AMM may be of independent interest.
\end{abstract}
\section{Introduction}
Given a matrix $A \in \R^{n \times d}$, a vector $b \in \R^n$, and a parameter $\lambda \ge 0$, the ridge regression problem is defined as:
\begin{equation*}
    \min_x \opnorm{Ax - b}^2 + \lambda\opnorm{x}^2.
\end{equation*}
Throughout the paper, we assume $n \le d$, and that $x^*$ is the optimal solution for the problem. Let $\opt$ be the optimal value for the above problem. Earlier work \cite{chowdhury2018iterative}  gives an iterative algorithm using so-called subspace embeddings. The following theorem states their results when their algorithm is run for 1 iteration. Note that their algorithm is more general and when run for $t$ iterations, the error is proportional to $\varepsilon^t$.
\begin{theorem}[Theorem~1 of \citet{chowdhury2018iterative}]
\label{thm:restatement-iterative-algorithm-one-iteration}
Given $A \in \R^{n \times d}$, let $V \in \R^{d \times n}$ be an orthonormal basis for the rowspace of matrix $A$. If $S \in \R^{m \times d}$ is a matrix which satisfies
\begin{equation}
    \opnorm{\T{V}\T{S}SV - I_n} \le \varepsilon/2,
    \label{eqn:requirement-chowdhury}
\end{equation}
then $\tilde{x} = \T{A}(A\T{S}S\T{A} + \lambda I_n)^{-1}b$ satisfies
\begin{equation*}
    \opnorm{\tilde{x} - x^*} \le \varepsilon\opnorm{x^*}.
\end{equation*}
\end{theorem}
A matrix $S$ which satisfies \eqref{eqn:requirement-chowdhury} is called an $\varepsilon/2$ subspace embedding for the column space of $V$, since for any $y$ in colspan($V$), we have
$
	(1-\varepsilon/2)\opnorm{y}^2 \le \opnorm{Sy}^2 \le (1+\varepsilon/2)\opnorm{y}^2
$. We also frequently drop the term ``column space'' and say $S$ is an $\varepsilon/2$ subspace embedding for the matrix $V$ itself.

There are many \emph{oblivious} and \emph{non-oblivious} constructions of subspace embeddings. As the name suggests, \emph{oblivious} subspace embedding (OSE) constructions do not depend on the matrix $V$ that is to be embedded. OSEs specify a distribution $\calS$ such that for any arbitrary matrix $V$, a random matrix $\bS$ drawn from the distribution $\calS$ is an $\varepsilon$ subspace embedding for $V$ with probability $\ge 1 - \delta$. On the other hand, non-oblivious constructions compute a distribution $\calS$ that depends on the matrix $V$ that is to be embedded. See the survey by \citet{dw-sketching} for an overview. 

In many cases, such as in streaming, it is important that the sketch used is \emph{oblivious}, since matrix-dependent subspace embedding constructions may need to read the entire input matrix first. Oblivious sketches also allow turnstile updates to the matrix $A$ in a stream. In the turnstile model of streaming, we receive updates of the form $((i,j), v)$ which update $A_{i,j}$ to $A_{i,j} + v$. In our paper we focus on algorithms for ridge regression that use \emph{oblivious} sketching matrices.

To satisfy \eqref{eqn:requirement-chowdhury}, using CountSketch \cite{low-rank-approximation-regression-input-sparsity, meng2013low}, we can obtain a sketching dimension of $m = O(n^2/\varepsilon^2)$ for which the matrix $S\T{A}$ can be computed in $O(\nnz(A))$ time, where $\nnz(A)$ denotes the number of nonzero entries in the matrix $A$. Using OSNAP embeddings \cite{osnap, cohen2016nearly}, we can obtain a sketching dimension of $m = O(n^{1+\gamma}\log(n)/\varepsilon^2)$ for which the matrix $S\T{A}$ can be computed in time $O(\nnz(A)/\gamma\varepsilon)$. For $\gamma = O(1/\log(n))$, we have $m = O(n\log(n)/\varepsilon^2)$ with $S\T{A}$ that can be computed in time $O(\nnz(A)\log(n)/\varepsilon)$. We can see that there is a tradeoff between CountSketch and OSNAP --- one has a smaller sketching dimension while the other is faster to apply to a given matrix. If $t_{S\T{A}}$ is the time required to compute $S\T{A}$, then $\tilde{x}$ in Theorem~\ref{thm:restatement-iterative-algorithm-one-iteration} can be computed in time $O(\nnz(A) + t_{S\T{A}} + mn^{\omega - 1} + n^{\omega})$ where $\omega$ is the matrix multiplication constant. Thus it is important to have both a small $t_{S\T{A}}$ and small $m$ to obtain fast running times.

When allowed $O(\log(1/\varepsilon))$ passes over the input matrix $A$, the algorithm of \citet{chowdhury2018iterative} produces an $\varepsilon$ relative error solution using only a constant, say $1/2$ subspace embedding. When only $O(1)$ passes are allowed over the input, their algorithm requires a $\delta = f(\varepsilon)$ subspace embedding to obtain $\varepsilon$ error solutions. As seen above this leads to either a high value of $m$ or a high value of $t_{S\T{A}}$.

We show that we only need a simpler Approximate Matrix Multiplication (AMM) guarantee, along with a constant subspace embedding, instead of requiring $S$ to be an $\varepsilon/2$ subspace embedding.

\begin{definition}[AMM]
	Given matrices $A$ and $B$ of appropriate dimensions, a matrix $S$ satisfies the  $\varepsilon$-AMM property for $(A,B)$ if
	\begin{align*}
		\frnorm{\T{A}\T{S}SB - \T{A}B} \le \varepsilon\frnorm{A}\frnorm{B}.
	\end{align*}
\end{definition}
We now state the guarantees of our algorithm (Algorithm~\ref{alg:main-alg}) for $1$ iteration, requiring $2$ passes over the matrix $A$.
\begin{theorem}
If $\bS$ is a random matrix such that for any fixed $d \times n$ orthonormal matrix $V$ and a vector $r$, with probability $\ge 9/10$,
\begin{equation*}
    \opnorm{\T{V}\T{\bS}\bS V - I_n} \le 1/2
\end{equation*}
and
\begin{equation}
    \opnorm{\T{V}\T{\bS}\bS Vr - r} \le (\varepsilon/2\sqrt{n}) \frnorm{V}\opnorm{Vr} = (\varepsilon/2) \opnorm{r},
\end{equation}
then $\tilde{x} = \T{A}(A\T{\bS}\bS\T{A} + \lambda I_n)^{-1}b$ satisfies $\opnorm{\tilde{x} - x^*} \le \varepsilon\opnorm{x^*}$ with probability $\ge 9/10$.
\label{thm:intro-one-iteration-statement}
\end{theorem}
We show that the OSNAP distribution satisfies both of these two properties with a sketching dimension  of $r = O(n\log(n) + n/\varepsilon^2)$ and with $S\T{A}$ that can be computed in $O(\nnz(A) \cdot \log(n))$ time. Note that our algorithm (Algorithm~\ref{alg:main-alg}) is also more general, and when run for $t$ iterations, the error is proportional to $\varepsilon^t$. Our algorithm differs from that of \citet{chowdhury2018iterative} in that our algorithm needs a fresh sketching matrix in each iteration whereas their algorithm only needs one sketching matrix across iterations. 

Many natural problems in the streaming literature have been studied specifically with $2$ passes \cite{chen2021near,konrad2021two,assadi2020near,brody2011streaming}. Also in the case of federated learning, where minimizing the number of rounds of communication is important \cite{park2021few}, the smaller sketch sizes required by our algorithm (Algorithm~\ref{alg:main-alg}) gives an improvement over the algorithm of \citet{chowdhury2018iterative}.

We can also bound the cost of $\tilde{x}$ computed by our algorithm. For any $x \in \R^d$, let $\cost(x) = \opnorm{Ax - b}^2 + \lambda\opnorm{x}^2$. Bounds on $\opnorm{\tilde{x} - x^*}$ also let us obtain an upper bound on $\cost(\tilde{x})$. It can be shown that for any vector $x$, $\cost(\tilde{x}) = \opt + \opnorm{A(x^* - \tilde{x})}^2 + \lambda\opnorm{x^* - \tilde{x}}^2$. Thus, $\opnorm{\tilde{x} - x^*} \le \varepsilon\opnorm{x^*}$ implies that $\cost(\tilde{x}) = \opt + (\sigma^2 + \lambda)\varepsilon^2\opnorm{x^*}^2 \le (1 + (1 + \sigma^2/\lambda)\varepsilon^2)\opt$. Throughout the paper, we are most interested in the case $\sigma^2 \ge \lambda$, as it is when $\cost(\tilde{x})$ could be much higher than $\opt$. Setting $\varepsilon = O(\sqrt{\delta\lambda/\sigma^2})$, we obtain that the solution $\tilde{x}$ returned by Theorem~\ref{thm:intro-one-iteration-statement} is a $1+\delta$ approximation.

We also show that our algorithm can be used to obtain approximate solutions to Kernel Ridge Regression with a polynomial kernel. We show that instantiating the construction of \citet{ahle2020oblivious} with appropriate sketching matrices gives a fast way to apply sketches, satisfying the subspace embedding and AMM properties, to the matrix $\phi(A)$, where the $i$-th row of the matrix $\phi(A)$ is given by $A_{i*}^{\otimes p}$.

\subsection{Lower bounds for Ridge Regression}
It can be seen that the optimal solution $x^* =\T{A}(A\T{A} + \lambda I_n)^{-1}b$. Our algorithm, for one iteration, is simply to compute $\tilde{x} = \T{A}(A\T{\bS}\bS\T{A}+ \lambda I_n)^{-1}b$ for a matrix $\bS$ that satisfies the requirements in Theorem~\ref{thm:intro-one-iteration-statement}. All the algorithm does is substitute the expensive matrix product $A\T{A}$, which can take $O(n \cdot \nnz(A))$ time to compute, with the matrix product $A\T{\bS}\bS \T{A}$, which only takes $t_{\bS \T{A}} +mn^{\omega - 1}$ time to compute. Thus, constructing ``good'' distributions for which $\tilde{x}$ is a $1+\varepsilon$ approximation seems to be the most natural way to obtain fast algorithms for ridge regression. As discussed previously, OSNAP matrices with $m = O(n\log(n) + n\sigma^2/\lambda\varepsilon)$ and having near-optimal $t_{\bS\T{A}} = \tilde{O}(\nnz(A))$ can be used to compute a solution $\tilde{x}$ that is a $1+\varepsilon$ approximation. We show that for a large class of nice-enough distributions over $m \times d$ matrices $\calS$, if $\bS \sim \calS$ satisfies that $\tilde{x} = \T{A}(A\T{\bS}\bS \T{A} + \lambda I)^{-1}b$ is a $1+\varepsilon$ approximation with high probability, then $r = \Omega(n\sigma^2/\lambda\varepsilon)$. This shows that OSNAP matrices have both a near-optimal sketching dimension $r$ and near-optimal time $t_{\bS\T{A}}$. We show the lower bound by showing that for any ``nice'' distribution  $\calS$ for which $\tilde{x}$ is a $1+\varepsilon$ approximation with high probability, the distribution must also satisfy an Approximate Matrix Multiplication (AMM) guarantee, i.e., for any matrix $B$, for $\bS \sim \calS$, $\frnorm{\T{B}\T{\bS}\bS B - \T{B}B}$ must be small with high probability. We then show a lower bound on $m$ for any distribution $\calS$ which satisfies the AMM  guarantee. Here we demonstrate our techniques in the simple case of $n = 1$. Without loss of generality, we assume $\lambda = 1$.

Consider the ridge regression problem
$
\min_x	(\T{a}x - b)^2 +  \opnorm{x}^2,
$
where $a$  is an arbitrary $d$-dimensional vector. We have $\tilde{x} = a(\T{a}\T\bS{\bS}a + 1)^{-1}b$ and 
\begin{equation*}
    \text{cost}(\tilde{x}) = \left(\frac{\opnorm{a}^2b}{\opnorm{{\bS}a}^2 + 1} - b\right)^2+\frac{\opnorm{a}^2}{(\opnorm{\bS a}^2 + 1)^2}b^2
\end{equation*}
 whereas $\opt = {b^2}/{(\opnorm{a}^2+1)}$. For $\opnorm{a} \ge 100/\sqrt{\varepsilon}$, it turns out that unless $(1-\sqrt{\varepsilon}/\opnorm{a})\opnorm{a}^2 \le \opnorm{\bS a}^2 \le (1 + \sqrt{\varepsilon}/\opnorm{a})\opnorm{a}^2$, we will have $\cost(\tilde{x}) \ge (1 + \varepsilon/2)\opt$. Thus for $\tilde{x}$ to be a $1+\varepsilon/2$ approximation with probability $\ge 99/100$ for any arbitrary $a$, it must be the case that with probability $\ge 99/100$, $|\T{a}a - \T{a}\T{\bS}\bS a| = |\opnorm{a}^2 - \opnorm{\bS a}^2| \le (\sqrt{\varepsilon}/\opnorm{a})\opnorm{a}^2$ i.e., $\bS$ must satisfy the AMM property with parameter $\sqrt{\varepsilon}/\opnorm{a}$. We show an $\Omega(1/\delta^2)$ lower bound for any distribution which satisfies the $\delta$-AMM property, which gives a lower bound of $\Omega(\opnorm{a}^2/\varepsilon)$ for ridge regression for $n=1$.
 
 For the case of general $n$, we show that any ``nice'' distribution $\calS$ that gives $1+\varepsilon$ approximate solutions for ridge regression must satisfy the  $\sqrt{\varepsilon/n\sigma^2}$-AMM guarantee,  which by using the lower bound for AMM, gives an $\Omega(n\sigma^2/\varepsilon)$ lower bound for ridge regression.
 
To prove the lower bound, we crucially use the fact that the sketching distribution $\calS$ must satisfy that $\tilde{x}$ is a $1+\varepsilon$ approximation for \emph{any} particular ridge regression problem instance $(A,b)$ with high probability.
 \subsection{Lower bounds for AMM}
 We prove the following lower bound for oblivious sketching matrices that give AMM guarantees.
 \begin{theorem}[Informal]
 If $\calS$ is a distribution over $m \times d$ matrices such that for any $n \times d$ matrix $A$, $\bS \sim \calS$ satisfies with probability $\ge 99/100$, that
 \begin{equation*}
     \frnorm{A\T{\bS}\bS\T{A} - A\T{A}} \le \delta\frnorm{A}\frnorm{\T{A}},
 \end{equation*}
 for $\delta \le c/\sqrt{n}$, then $m = \Omega(1/\delta^2)$ where $c > 0$ is a small enough universal constant.
 \end{theorem}
 To the best of our knowledge, this is the first tight lower bound on the dimension of oblivious sketching matrices for AMM. The lower bound is tight up to constant factors as the CountSketch distribution with $m = O(1/\delta^2)$ rows has the above property. Note that for $\delta = \varepsilon/n$, the distribution $\calS$ as in the above theorem satisfies that for any $d \times n$ orthonormal matrix $V$, with probability $\ge 99/100$,
 \begin{equation*}
     \frnorm{\T{V}\T{\bS}\bS V - I_n} \le (\varepsilon/n)\frnorm{V}^2 = \varepsilon.
 \end{equation*}
 Thus, a distribution $\calS$ that has the  $\varepsilon/n$-AMM property also has the  $\varepsilon$-subspace embedding property. \citet{nelson2014lower} gives an $\Omega(n/\varepsilon^2)$ lower bound for such distributions, thus giving an  $\Omega(1/(\delta^2n))$ lower bound for $\delta$-AMM for small enough $\delta$. The above theorem gives a stronger $\Omega(1/\delta^2)$ lower bound.
 
 We now give a brief overview of our proof for $n = 1$. Consider $a \in \R^d$ to be a fixed unit vector and let $\calS$ be a distribution supported on $r \times d$ matrices as in the above theorem. Then we have $\Pr_{\bS \sim \calS}[|\T{a}\T{\bS}\bS a - 1| \le \delta] \ge 0.99$. Let $U\Sigma\T{V}$ be the singular value decomposition of $\bS$ with $\Sigma \in \R^{r \times r}$. Without loss of generality, we can assume that $\T{V}$ is independent of $\Sigma$ and that $\T{V}$ is a uniformly random orthonormal matrix. This follows from the fact that if $\bS$ is an \emph{oblivious} AMM sketch, then $\bS \bQ$ is also an \emph{oblivious} AMM sketch, where $\bQ$ is a uniformly random $d \times d$ orthogonal matrix independent of $\bS$. Thus, we have $\Pr_{\Sigma, \T{V}}[|\T{a}V\Sigma^2\T{V}a - 1| \le \delta] \ge 0.99$, where $\Sigma$ and $\T{V}$ are random matrices that correspond to the AMM sketch $\bS$, as described. 
 
 \citet{jiang2017distances} show that if $m = o(d)$, then the total variation distance between $\T{V}a$ and $(1/\sqrt{d})\bg$ is small, where $\bg$ is an $m$ dimensional vector with independent Gaussian entries. 
 Thus, we obtain that 
$
  \Pr_{\Sigma, \bg}[|(1/d)\T{\bg}\Sigma^2\bg - 1| \le \delta]
  = \Pr_{\Sigma, \bg}[|(1/d)\sum_{i=1}^m \sigma_i^2\bg_i^2 - 1| \le \delta] \ge 0.95.   
$

 If $\sum_{i=1}^m  \sigma_i^2 \le d/200$, then $(1/d)\sum_{i=1}^m \sigma_i^2\bg_i^2 \le 1/2$ with probability $\ge 0.99$ by Markov's  inequality. So, $\Pr_{\Sigma}[\sum_{i=1}^m \sigma_i^2 \le d/200]$ must be small. On the other hand, $\Var{((1/d)\sum_i \sigma_i^2\bg_i^2)} = (2/d^2)\sum_{i=1}^m \sigma_i^4$. Thus for $(1/d)\sum_{i=1}^m \sigma_i^2\bg_i^2$ to concentrate in the interval $(1-\delta, 1+\delta)$, we would expect $\sqrt{\Var{((1/d)\sum_i \sigma_i^2\bg_i^2)}} \approx \delta$, which implies $(2/d^2)\sum_{i=1}^m \sigma_i^4 \approx \delta^2$. Thus, with a reasonable probability, it must be simultaneously true that $d/200 \le \sum_{i=1}^m\sigma_i^2$ and $\sum_{i=1}^m \sigma_i^4 \approx d^2\delta^2/2$. Then, 
 \begin{equation*}
     d^2/(200)^2 \le \left(\sum_{i=1}^m \sigma_i^2\right) \le m\sum_{i=1}^m \sigma_i^4 \approx md^2\delta^2/2, 
 \end{equation*}
 thus obtaining $m \gtrsim \Omega(1/\delta^2)$. We extend this proof idea to the general case of $n \ge 1$.
 
 Non-asymptotic upper bounds on the total variation (TV) distance between Gaussian matrices and sub-matrices of random orthogonal matrices obtained in recent works \cite{jiang2017distances, product-of-gaussians} let us replace the matrices that are harder to analyze with Gaussian matrices in our proof of the lower bound for AMM. We believe this technique could be helpful in proving tight lower bounds for other types of sketching guarantees. 
 
 \subsection{Other Contributions}
 We also show lower bounds on the communication complexity  for approximating the \emph{optimal value} of the ridge regression problem, which is a different, but related problem, to the computation of $1+\varepsilon$ approximate \emph{solutions} which is studied in this paper. We obtain an $\Omega(1/\varepsilon^2)$ bit lower bound for ridge regression when $\sigma^2/\lambda \approx 1$. The hard instance is a two-party communication game on a $2 \times d$ matrix, for which one party has the first row and the other party has the second row. Surprisingly, if each party has $d/2$ columns of the design matrix $A = [A_1\, A_2]$, then they can compute the exact optimal value if the first party communicates $A_1\T{A_1}$ to the second party using $O(n^2)$ words of communication. This suggests that the turnstile streaming setting is harder than the column arrival setting for streaming algorithms for ridge regression. We stress that our algorithm to compute a $(1+\varepsilon)$-approximate solution to ridge regression works in the turnstile streaming setting by maintaining $\bS \T{A}$ in a stream. Nevertheless, to output the $d$ dimensional solution $\tilde{x}$, our algorithm needs to compute \emph{one} matrix-vector product with $\T{A}$ at the end of the stream, necessitating a second pass over the stream. 
 
 In contrast, we obtain $\Omega(d)$ bit communication complexity lower bounds for the Lasso and square-root Lasso objectives, even for computing $1+c$ approximations to the optimum value for a small enough constant $c$. Fast algorithms for these objectives seem to be harder to find and as the lower bounds indicate, there may not be sketching-based algorithms for these problems.
 
 We defer most of the proofs to the supplementary material. We also include an experiment there comparing the time required by our algorithm to the time required to compute the exact solution for ridge regression.
\section{Preliminaries}
For $n \in \mathbb{Z}$, $[n]$ denotes the set $\set{1,\ldots,n}$. Given a matrix $A \in \R^{n \times d}$, let $\nnz(A)$ denote the number of nonzero entries in $A$. For the matrix $A$, $\frnorm{A}$ denotes the Frobenius norm $(\sum_{i,j} A_{i,j}^2)^{1/2}$ and $\opnorm{A}$ denotes the spectral (operator) norm $\max_{x \ne 0}\opnorm{Ax}/\opnorm{x}$. When there is no ambiguity, throughout the paper we use $\sigma$ to denote $\opnorm{A}$. Let $A = U\Sigma\T{V}$ be the Singular Value Decomposition (SVD) with $U \in \R^{n \times \rho}$, $\Sigma \in \R^{\rho \times \rho}$, and $V \in \R^{d \times \rho}$, where $\rho = \text{rank}(A)$. For arbitrary matrices $M,N$, the symbol $t_{MN}$ denotes the time required to compute the product $MN$.

We use uppercase symbols $A,U,V,S,\ldots$ to denote matrices and lowercase symbols $a,b,u,v,\ldots$ to denote vectors. For a matrix $A$, $A_{i*}$ ($A_{*i}$) denotes the $i$-th row (column). We use boldface symbols $\bU,\bV,\bS,\br,\ldots$ to stress that these objects are random and are explicitly sampled from an appropriate distribution.

\begin{definition}[Approximate Matrix Multiplication]
    Given an integer $d$, we say that an $m \times d$ random matrix $\bS$ has the  $(\varepsilon, \delta)$-AMM property if for any matrices $A$ and $B$ with $d$ rows, we have that
    \begin{equation*}
        \frnorm{\T{A}\T{\bS}\bS B - \T{A}B} \le \varepsilon\frnorm{A}\frnorm{B}.
    \end{equation*}
    with probability $\ge 1-\delta$ over the randomness of $\bS$.
\end{definition}
We usually drop $\delta$ from the notation by picking it to be a small enough constant. 
\begin{definition}[Oblivious Subspace Embeddings]
    Given an integer $d$, an $m \times d$ random matrix $\bS$  is an $(\varepsilon,\delta)$-OSE for $n$-dimensional subspaces if for any arbitrary $d \times n$ matrix $A$, with probability $\ge 1-\delta$, simultaneously for all vectors $x$,
    \begin{equation*}
        \opnorm{\bS Ax}^2\in (1 \pm \varepsilon)\opnorm{Ax}^2.
    \end{equation*}
\end{definition}
For both OSEs and distributions satisfying the  $(\varepsilon,\delta)$-AMM property, two major parameters of importance are the size of the sketch ($m$), and the time to compute $\bS A$ ($t_{\bS A}$). See \citet{dw-sketching} and the references therein for several OSE constructions and their corresponding parameters.

\section{Iterative Algorithm for Ridge Regression}\label{sec:main-alg}
The following theorem describes the guarantees of the solution $\hat{x}$ returned by Algorithm~\ref{alg:main-alg}.
\begin{theorem}
	If Algorithm~\ref{alg:main-alg} samples independent sketching matrices $\bS_j \in \R^{m \times d}$ for all $j \in [t]$ satisfying the properties
	\begin{enumerate}
		\item with probability $\ge 1 - 1/(20t)$, for all vectors $x$, $\opnorm{\bS_j \T{A}x}^2 \in (1 \pm 1/2)\opnorm{\T{A}x}^2$, and
		\item for all arbitrary matrices $M,N$, with probability $\ge  1 - 1/(20t)$, 
		\begin{align*}
			\frnorm{\T{M}\T{\bS_j}\bS_j N - \T{M}N} \le \sqrt{\epsilon/4n}\frnorm{M}\frnorm{N},
		\end{align*}
	\end{enumerate}
	then with probability $\ge 9/10$, $\opnorm{\hat{x} - x^*} \le (\sqrt{\varepsilon})^t\opnorm{x^*}$ and further $\cost(\hat{x}) \le (1 + (\sigma^2/\lambda + 1)\varepsilon^t)\opt$.
	\label{thm:main-thm}
\end{theorem}
We prove a few lemmas which give intuition about the algorithm before proving the above theorem.
\begin{algorithm}
\caption{\textsc{RidgeRegression}}
\label{alg:main-alg}
\begin{algorithmic}
\STATE {\bfseries Input:} $A \in \R^{n \times d}, b \in \R^d, t \in \mathbb{Z}, \varepsilon,\lambda > 0$
\STATE {\bfseries Output:} $\hat{x} \in \R^d$
\STATE $b^{(0)} \gets b$, $\tilde{x}^{(0)} \gets 0_d$, $y^{(0)} \gets 0_n$
\FOR{$j=1,\ldots,t$}
\STATE $b^{(j)} \gets b^{(j-1)} - \lambda y^{(j-1)} - A\tilde{x}^{(j-1)}$
\STATE $\bS_j \gets $ $1/2$ subspace embedding for the rowspace of $A$ and has the  $\sqrt{\varepsilon/4n}$ AMM property
\STATE $y^{(j)} \gets (A \T{\bS_j}\bS_j\T{A} + \lambda I)^{-1}b^{(j)}$
\STATE $\tilde{x}^{(j)} \gets \T{A}y^{(j)}$
\ENDFOR
\STATE $\hat{x} \gets \sum_{j=1}^t \tilde{x}^{(j)}$
\STATE {\bfseries return} $\hat{x}$
\end{algorithmic}
\end{algorithm}

After $i-1$ iterations of the algorithm, $\sum_{j=1}^{i-1} \tilde{x}^{(j)}$ is the estimate for the optimum solution $x^*$. At a high level, in the $i$-th iteration, the algorithm is trying to compute an approximation to the difference $x^* - \sum_{j=1}^{i-1}\tilde{x}_j$ by computing an approximate solution to the problem
\begin{equation*}
    \min_x \opnorm{A(x + \sum_{j=1}^{i-1}\tilde{x}_j) - b}^2 + \lambda\opnorm{x + \sum_{j=1}^{i-1}\tilde{x}_j}^2.
\end{equation*}

Let $x^{*(j)} = \T{A}(A\T{A} + \lambda I)^{-1}b^{(j)}$. The following lemma shows that the solution to the above problem is $x^{*(i)}$.
\begin{lemma}
For all $i$, $x^* = x^{*(i)}+\sum_{j=1}^{i-1}\tilde{x}^{(j)}$.
\label{lma:i-th-iteration-objective}
\end{lemma}
\begin{proof}
Let $f(x) = \opnorm{A(x + \sum_{j=1}^{i-1}\tilde{x}^{(j)}) - b}^2 + \lambda\opnorm{x + \sum_{j=1}^{i-1}\tilde{x}^{(j)}}^2$ and $z$ be the solution realizing $\min_x f(x)$. We have $\nabla_x f(x)\,|_{x = z} = 0$ giving
$
    z = (\T{A}A + \lambda I)^{-1}(\T{A}b - (\T{A}A + \lambda I)\sum_{j=1}^{i-1}\tilde{x}^{(j)}).
$

Noting that $\tilde{x}^{(j)} = \T{A}y^{(j)}$ for all $j$ and that for all $i$, $b^{(i)} = b - \lambda \sum_{j=1}^{i-1}y^{(j)} - \sum_{j=1}^{i-1}A\tilde{x}^{(j-1)}$, we obtain that $z = (\T{A}A + \lambda I)^{-1}\T{A}b^{(i)}$. Now using the matrix identity $(\T{A}A + \lambda I)^{-1}\T{A} = \T{A}(A\T{A} + \lambda I)^{-1}$, we get $z = x^{*(i)}$ is the optimal solution to $\min_x f(x)$.

As $x^{*(i)}$ is the optimal solution, it is also clear that $x^* = x^{*(i)} + \sum_{j=1}^{i-1}\tilde{x}^{(j)}$ since otherwise $x^*$ is not the optimal solution for the original ridge regression problem, which is a contradiction.
\end{proof}
So by the end of the $(j-1)$-th iteration, the estimate to $x^*$ is off by $x^{*(j)}$. The algorithm is approximating $x^{*(j)} = \T{A}(A\T{A} + \lambda I)^{-1}b^{*(j)}$ with $\tilde{x}^{(j)} = \T{A}(A\T{\bS_j}{\bS_j}\T{A} + \lambda I)^{-1}b^{*(j)}$.
The following lemma gives the error of this approximation assuming that the sketching matrix $\bS_j$ has both the subspace embedding and AMM properties. This is the part where our proof differs from that of the proof of \citet{chowdhury2018iterative}.
\begin{lemma}
	If $\bS_j$ is drawn from a distribution such that for any fixed matrix $\T{A}$, $\bS_j$ is a $1/2$ subspace embedding with probability $1-\delta$ and for any fixed matrices $M,N$, with probability $1-\delta$, 
	\begin{equation*}
		\frnorm{\T{M}\T{\bS_j}\bS_jN - \T{M}N} \le \sqrt{\varepsilon/n}\frnorm{M}\frnorm{N},
	\end{equation*}
	then with probability $\ge 1 - 2\delta$, $\opnorm{x^{*(j)} - \tilde{x}^{(j)}} \le (2\sqrt{\varepsilon})\opnorm{x^{*(j)}}$.
	\label{lma:j-th-iteration-approximation-guarantee}
\end{lemma}
\begin{proof}
Let $A = U\Sigma \T{V}$ be the singular value decomposition of $A$. We have $x^{*(j)} = V\Sigma(I + \Sigma^2)^{-1}\T{U}b^{(j)}$. By using $(I+\Sigma^2)^{-1} = \Sigma^{-1}(I+\Sigma^{-2})^{-1}\Sigma^{-1}$, we get $x^{*(j)}= V(I+\Sigma^{-2})^{-1}\Sigma^{-1}\T{U}b^{(j)}$. Let $v^{(j)}= (I+\Sigma^{-2})^{-1}\Sigma^{-1}\T{U}b^{(j)}$ which gives $x^{*(j)} = Vv^{(j)}$. 

Similarly, $\tilde{x}^{(j)} = V(\T{V}\T{\bS_j}\bS_j V + \Sigma^{-2})^{-1}\Sigma^{-1}\T{U}b^{(j)}$. Writing $\T{V}\T{\bS_j}\bS_j V = I_n + E$, we have
\begin{align*}
	&\tilde{x}^{(j)} = V(I+\Sigma^{-2} + E)^{-1}\Sigma^{-1}\T{U}b^{(j)}\\
	&= V(I + (I+\Sigma^{-2})^{-1}E)^{-1}(I+\Sigma^{-2})^{-1}\Sigma^{-1}\T{U}b^{(j)}\\
	&= V(I+(I+\Sigma^{-2})^{-1}E)^{-1}v^{(j)}.
\end{align*}
As $\opnorm{E} \le 1/2$, the inverse $(I + (I+\Sigma^{-2})^{-1}E)^{-1}$ is well-defined.
Since the matrix $V$ has orthonormal columns, $\opnorm{\tilde{x}^{(j)} - x^{*(j)}} = \opnorm{(I+(I+\Sigma^{-2})^{-1}E)^{-1}v^{(j)} - v^{(j)}}$. Let $(I+(I+\Sigma^{-2})^{-1}E)^{-1}v^{(j)} = v^{(j)} + \Delta$  and we have $v^{(j)} = v^{(j)} + (I+\Sigma^{-2})^{-1}Ev^{(j)} + (I+(I+\Sigma^{-2})^{-1}E)\Delta$ which implies $(I+(I+\Sigma^{-2})^{-1}E)\Delta = - (I+\Sigma^{-2})^{-1}Ev^{(j)}$. Finally,
\begin{align*}
	(1/2)\opnorm{\Delta} &\le \sigma_{\min}(I+(I+\Sigma^{-2})^{-1}E)\opnorm{\Delta}\\
	&\le \opnorm{(I+(I+\Sigma^{-2})^{-1}E)\Delta}\\
	&= \opnorm{(I+\Sigma^{-2})^{-1}Ev^{(j)}} \le \opnorm{Ev^{(j)}}.
\end{align*}
which gives $\opnorm{x^{*(j)} - {\tilde x}^{(j)}} = \opnorm{V\Delta} \le 2\opnorm{Ev^{(j)}}$. If the matrix $\bS_j$ has a $\sqrt{\varepsilon/n}$-AMM property i.e., 
\begin{align*}
\opnorm{\T{V}\T{\bS_j}\bS_j Vv^{(j)} - \T{V}Vv^{(j)}} &\le \sqrt{\varepsilon/n}\frnorm{V}\opnorm{v^{(j)}}\\
&= \sqrt{\varepsilon}\opnorm{v^{(j)}},
\end{align*}
 we have $\opnorm{Ev^{(j)}} \le \sqrt{\varepsilon}\opnorm{Vv^{(j)}}$ and that $\opnorm{x^{*(j)} - \tilde{x}^{(j)}} \le 2\sqrt{\varepsilon}\opnorm{v^{(j)}} = 2\sqrt{\epsilon}\opnorm{x^{*(j)}}$.
\end{proof}

\begin{proof}[Proof of Theorem~\ref{thm:main-thm}]
	By a union bound, with probability $\ge 9/10$, in all $t$ iterations, we can assume that the matrices $\bS_j$ have both the subspace embedding property for the  column space of $\T{A}$, as well as the AMM property for $V$ and $v^{(j)}$.
	
From Lemma~\ref{lma:i-th-iteration-objective}, $\opnorm{\hat{x} - x^*} = \opnorm{\tilde{x}^{(t)} + \sum_{i=1}^{t-1}\tilde{x}^{(i)} - x^*} = \opnorm{\tilde{x}^{(t)} - x^{*(t)}} \le (\sqrt{\varepsilon})\opnorm{x^{*(t)}}$.
We also have  
\begin{equation*}
x^* = x^{*(j-1)} + \sum_{i=1}^{j-2}\tilde{x}^{(i)} = x^{*(j)} + \sum_{i=1}^{j-1}\tilde{x}^{(i)}    
\end{equation*}
which implies $x^{*(j)} = x^{*(j-1)} - \tilde{x}^{(j-1)}$ and therefore, $\opnorm{x^{*(j)}} = \opnorm{\tilde{x}^{(j-1)} - x^{*(j-1)}} \le \sqrt{\varepsilon}\opnorm{x^{*(j-1)}}$ for all $j$, where the last inequality follows from Lemma~\ref{lma:j-th-iteration-approximation-guarantee}. Now noting that $x^{*(1)} = x^*$, we obtain $\opnorm{\hat{x} - x^*} \le (\sqrt{\varepsilon})^t\opnorm{x^*}$ and using the Pythagorean theorem,
\begin{align*}
	\cost(\hat{x}) &\le \opt + (\sigma^2+\lambda)\opnorm{\hat{x} - x^*}^2\\
	&\le \opt + (\sigma^2 + \lambda)\varepsilon^t\opnorm{x^*}^2.
\end{align*}
As $\lambda\opnorm{x^*}^2 \le \opt$, we obtain the result.
\end{proof}
We now show that the OSNAP distribution has both the properties required by Algorithm~\ref{alg:main-alg}.
\subsection{Properties of OSNAP}
\citet{osnap} proposed OSNAP, an oblivious subspace embedding. OSNAP embeddings are parameterized by their number $m$ of rows and their sparsity $s$. Essentially, OSNAP is a random $m \times d$ matrix $\bS$, with each column having exactly $s$ nonzero entries at random locations. Each nonzero entry is $\pm 1/\sqrt{s}$ with probability $1/2$ each. They show that if the positions of the nonzero entries satisfy an ``expectation'' property and if the nonzero values are drawn from a $k$-wise independent distribution for a sufficiently large $k$, then $\bS$ is an OSE.
\begin{theorem}[Informal, \cite{osnap}]
If $m = O(n^{1+\gamma}\poly(\log(n),1/\varepsilon)/\varepsilon^2)$ and $s = O(1/\gamma\varepsilon)$, then OSNAP is an $\varepsilon$-OSE for $n$ dimensional spaces. Further, $t_{\bS A} = O(\nnz(A)/\gamma\varepsilon)$ for any $d \times n$ matrix $A$.
\end{theorem}
In the supplementary, we show that OSNAP with \emph{any} sparsity parameter $s$ and $m = \Omega(1/\varepsilon^2)$ has the $\varepsilon$-AMM property. We state our result as the following lemma.
\begin{lemma}
    OSNAP with $m = \Omega(1/\varepsilon^2\delta)$ and sparsity parameter $s \ge 1$ has the  $(\varepsilon,\delta)$-AMM property.
\end{lemma}
\subsection{Running times : Embedding vs Current Work}
As discussed in the introduction, the algorithm of \citet{chowdhury2018iterative} is better than ours when $O(\log(1/\varepsilon))$ passes over the matrix $A$ are allowed, as we require a fresh $1/2$ subspace embedding in each iteration and they require only one $1/2$ subspace embedding. However, our algorithm is faster when the algorithm is restricted to $t = O(1)$ passes over the input. We compare the running time of our algorithm with theirs when both algorithms are run only for $1$ iteration to obtain $1+\varepsilon$ approximate solutions. For ease of exposition, we consider the case when $\sigma^2/\lambda = O(1)$.

From Theorem~\ref{thm:restatement-iterative-algorithm-one-iteration}, the algorithm of \citet{chowdhury2018iterative} requires a $c\sqrt{\varepsilon}$ subspace embedding to output a $1+\varepsilon$ approximation to ridge regression. By applying a sequence of CountSketch and OSNAP sketches, we can obtain a $c\sqrt{\varepsilon}$ embedding with $m = n\poly(\log(n))/\varepsilon$ and $t_{S\T{A}} = O(\nnz(A) + n^3\poly(\log(n))/\sqrt{\varepsilon})$ or by directly applying OSNAP, we obtain $m = n\poly(\log(n))/\varepsilon$ and $t_{S\T{A}} = O(\nnz(A)\poly(\log(n))/\sqrt{\varepsilon})$.

From Theorem~\ref{thm:main-thm}, our algorithm needs a random matrix that has the $1/2$ subspace embedding property and the  $c\sqrt{\varepsilon/n}$-AMM property to compute a $1+\varepsilon$ approximation. OSNAP with $m = O(n/\varepsilon + n\poly(\log(n)))$ and $s = O(\poly(\log(n)))$ has this property giving $t_{S\T{A}} = O(\nnz(A)\poly(\log(n)))$.

Finally, the total time to compute $\tilde{x}$ is \[O(t_{S\T{A}} + mn^{\omega-1} + n^{\omega}),\]where $\omega < 3$ denotes the matrix multiplication exponent. 
For the algorithm of \citet{chowdhury2018iterative}, depending on the sketching matrices used as described above, the total running time is either \[O(\nnz(A) + n^{3}\poly(\log(n))/\sqrt{\varepsilon} + n^{\omega}\poly(\log(n))/\varepsilon)\] or \[O(\nnz(A)\poly(\log(n))/\sqrt{\varepsilon} + n^{\omega}\poly(\log(n))/\varepsilon).\] 
For Algorithm~\ref{alg:main-alg} with $t=1$, the total running time is $O(\nnz(A)\poly(\log(n)) + n^{\omega}\poly(\log(n))/\varepsilon)$. 
Thus we have that when $\nnz(A) \approx n^{\omega}/\varepsilon$, our algorithm is asymptotically faster than their algorithm, as our running time does not have the $n^{3}$ term and $\nnz(A)/\sqrt{\varepsilon}$ terms. 
We note that although the fastest matrix multiplication algorithms are sometimes considered impractical, Strassen's algorithm is already practical for reasonable values of $n$, and gives $\omega < \log_2 7$. 
If we consider the algorithm of \citet{chowdhury2018iterative} using just the OSNAP embedding, our algorithm is faster by a factor of $1/\sqrt{\varepsilon}$, which could be substantial when $\varepsilon$ is small.

Even non-asymptotically, our result shows that we can replace the sketching matrix in their algorithm with a sketching matrix that is both sparser and has fewer rows, while still obtaining a $1+\varepsilon$ approximation. Both of these properties help the algorithm to run faster.

\section{Applications to Kernel Ridge Regression}\label{sec:krr}
A function $k : X \times X \rightarrow \R$ is called a positive semi-definite kernel if it satisfies the following two conditions:
	(i) For all $x,y \in X$, $k(x,y) = k(y,x)$, and
	(ii) for any finite set $S = \set{s_1,\ldots, s_t} \subseteq X$, the matrix $K = [k(s_i,s_j)]_{i,j \in [t]}$ is positive semi-definite.
Mercer's theorem states that a function $k(\cdot, \cdot)$ is a positive semi-definite kernel as defined above if and only if there exists a function $\phi$ such that for all $x,y \in X$, $k(x,y) = \T{\phi(x)}\phi(y)$. Many machine learning algorithms only work with inner products of the data points and therefore all such algorithms can work using the function $k$ directly instead of the explicit mapping $\phi$, which in principle could even be infinite dimensional, for example, as in the case of the Gaussian kernel. 
	

Let the rows of a matrix $A$ be the input data points $a_1,\ldots,a_n$, and let $\phi(A)$ denote the matrix obtained by applying the function $\phi$ to each row of the matrix $A$. The kernel ridge regression problem (see \citet{murphy2012machine} for more details) is defined as
\begin{equation*}
	c^* = \argmin_c \opnorm{\phi(A) \cdot c - b}^2 + \lambda \opnorm{c}^2.
\end{equation*}
We have that $c^* = \T{\phi(A)} (\phi(A) \cdot \T{\phi(A)} + \lambda I)^{-1}b$ and the value predicted for an input $x$ is given by $\T{\phi(x)}c^* = \T{\phi(x)}\T{\phi(A)}(\phi(A) \cdot \T{\phi(A)} + \lambda I)^{-1}b$. Letting $\beta = ({\phi}(A)\T{\phi(A)} + \lambda I)^{-1}b$ we have $\T{\phi(x)}c^* = \sum_{i}k(a_i, x)\beta_i$. Now, note that the $(i,j)$-th entry of the matrix $K := \phi(A) \cdot \T{\phi(A)}$ is given by $k(a_i, a_j)$ and therefore, to solve the kernel ridge regression problem, we do not need the explicit map $\phi( \cdot )$ and can work directly with the kernel function. Nevertheless, to construct the matrix $K$, we need to query the kernel function $k$ for $\Theta(n^2)$ pairs of inputs, which may be prohibitive if the kernel evaluation is slow.

Our result for ridge regression shows that if $\bS$ is a $1/2$ subspace embedding and gives a $\varepsilon/2\sqrt{n}$ AMM guarantee, then
\begin{equation*}
\tilde{c} = \T{\phi(A)} \cdot (\phi(A) \cdot \T{\bS} \bS \cdot \T{\phi(A)} + \lambda I)^{-1}b
\end{equation*}
satisfies $\|\tilde{c} - c^*\| \le \varepsilon\opnorm{c^*}$ and if $\tilde{\beta} := (\T{\phi}(A) \cdot \T{\bS} \bS \cdot \T{\phi(A)} + \lambda I)^{-1}b$, then for a new input $x$, the prediction on $x$ can be computed as $\sum_{i}k(a_i,x)\tilde{\beta}_i$. For polynomial kernels, $k(x,y) = \langle x,y\rangle^p$, given the matrix $A$, it is possible to compute $\bS \cdot \T{\phi(A)}$ for a random matrix $\bS$ that satisfies both the subspace embedding property and the AMM property, and hence obtain $\tilde{\beta}$ without computing the kernel matrix. The next theorem follows from the proofs of Theorems~1 and 3 of \citet{ahle2020oblivious}.

\begin{theorem}
\label{thm:nice-kernel-sketch}
	For all positive integers $n,d,p$, there exists a distribution on linear sketches $\Pi^p \in \R^{m \times d^p}$ parameterized by sparsity $s$ such that: if $m = \Omega(p/\varepsilon^{2})$ and any sparsity $s$, then $\Pi^p$ has the $\varepsilon$-AMM property, while if $m = \tilde{\Omega}(p^4n/\varepsilon^2)$ and $s = \tilde\Omega((p^{4}/\varepsilon^2\poly(\log(nd/\epsilon)))$, then $\Pi^p$ has the $\varepsilon$ subspace embedding property. Further, given any matrix $A \in \R^{n \times d}$, the matrix $\Pi^p \cdot \T{\phi(A)}$ for $\phi(x) = x^{\otimes p}$ can be computed in $\tilde{O}(pnm + ps \cdot \nnz(A))$ time. 
\end{theorem}
We show that the construction of \citet{ahle2020oblivious} gives the above theorem when $S_{\text{base}}$ is taken to be  TensorSketch and $T_{\text{base}}$ is taken to be OSNAP. To prove the theorem, we first prove a lemma which shows that the OSNAP distribution has the JL-moment property. For a random variable $\bX$, let $\|\bX\|_{L^t} := (\E[|\bX|^t])^{1/t}$.
\begin{definition}[JL-Moment Property]
	For every positive integer $t$ and parameters $\varepsilon,\delta \ge 0$, we say a random matrix $\bS \in \R^{m \times d}$ satisfies the $(\varepsilon, \delta, t)$-JL moment property if for any $x \in \R^{d}$ with $\opnorm{x}=1$,
	\begin{equation*}
		\|\opnorm{\bS x}^2 - 1\|_{L^t} \le \varepsilon\delta^{1/t}\ \text{and}\ \E[\opnorm{\bS x}^2] = 1.
	\end{equation*}
\end{definition}
\begin{lemma}
\label{lma:jl-moment-property}
    If $\bS$ is an OSNAP matrix with $m = \Omega(1/\delta\varepsilon^2)$ rows and any sparsity parameter $s \ge 1$, then $\bS$ has the $(\varepsilon, \delta, 2)$-JL moment property.
\end{lemma}

\begin{proof}[Proof of Theorem~\ref{thm:nice-kernel-sketch}]
Let $q = 2^{\lceil \log_2(p) \rceil}$. The construction of the sketch for polynomial kernels of \citet{ahle2020oblivious} uses two distributions of matrices $S_{\text{base}}$ and $T_{\text{base}}$. The proof of Theorem~1 of \citet{ahle2020oblivious} requires that the distributions $S_{\text{base}}$ and $T_{\text{base}}$ have the  $(\varepsilon/\sqrt{4q+2},\delta,2)$-JL moment property. We take $S_{\text{base}}$ to be TensorSketch and $T_{\text{base}}$ to be OSNAP. As Lemma~\ref{lma:jl-moment-property} shows, OSNAP with $m = \Omega(q/\delta\varepsilon^2)$ and any sparsity $s$ has the $(\varepsilon/\sqrt{4q+2},\delta,2)$-JL moment property. 

From Theorem~3 of \citet{ahle2020oblivious}, we also have that for $m = \tilde\Omega(p^4n/\varepsilon^2)$ and sparsity parameter $s = \tilde{\Omega}((p^4/\varepsilon^2) \cdot \poly(\log(nd/\varepsilon)))$, the sketch has the $\varepsilon$-subspace embedding property. The running time of applying the sketch to $\T{\phi(A)}$ also follows from the same theorem.
\end{proof}
Thus for the sketch to have both the $1/2$ subspace embedding property and the  $\varepsilon/\sqrt{4n}$ AMM property, we need to take $m = \tilde{\Omega}(p^4n + pn/\epsilon^2)$ and $s = \tilde{\Omega}(p^4\poly(\log(nd)))$.  The time to compute $\Pi^p \cdot \T{\phi(A)}$ is $\tilde{O}(p^5 \nnz(A) + p^5n^2 + p^2n/\varepsilon^2)$ and the time to compute $\tilde{\beta}$ is 
$\tilde{O}(p^5 \nnz(A) + p^5n^2 + p^2n^2/\varepsilon^2 + p^4n^{\omega} + pn^{\omega}/\epsilon^2)$, thereby obtaining a near-input sparsity time algorithm for polynomial kernel ridge regression.


\section{Lower bounds}
\label{sec:lowerbounds}
Dimensionality reduction, by multiplying the input matrix $A$ on the right with a random sketching matrix, seems to be the most natural way to speed up ridge regression. Recall that in our algorithm above, we show that we only need the sketching distribution to satisfy a simple AMM guarantee, along with being a constant factor subspace embedding, to be able to obtain a $1+\varepsilon$ approximation. We show that, in this natural framework, the bounds on the number of rows required for a sketching matrix we obtain are nearly optimal for all ``non-dilating'' distributions.

More formally, we show lower bounds in the restricted setting where for an oblivious random matrix $\bS$, the vector $\tilde{x} = \T{A}(A\T\bS{\bS}\T{A} + \lambda I)^{-1}b$ must be a $1+\varepsilon$ approximation to the ridge regression problem with probability $\ge 99/100$. We show that the matrix ${\bS}$ must at least have $m = \Omega(n\sigma^2/\lambda\varepsilon)$ rows if $\bS$ is ``non-dilating''.

\begin{definition}[Non-Dilating Distributions]
A distribution $\calS$ over $m \times d$ matrices is a Non-Dilating distribution if for all $d \times n$ orthonormal matrices $V$,
\begin{equation*}
    \Pr_{\bS \sim \calS}[\opnorm{\bS V} \le O(1)] \ge 99/100.
\end{equation*}
\end{definition}
Most sketching distributions proposed in previous work satisfy the property $\E[\T{V}\T{\bS}\bS V] = \T{V}V = I$. Thus the condition of non-dilation is not very restrictive. For example, a Gaussian distribution with $O(n)$ rows satisfies this condition, and other sketching distributions such as SRHT, CountSketch, and OSNAP with $O(n\log(n))$ rows all satisfy this condition with $O(1)$ replaced by at most $O(\log(n))$. Though we prove our lower bounds for non-dilating distributions with $O(1)$ distortion, the lower bounds also hold with distributions with $O(\log(n))$ distortion with at most an $O(\log(n))$ factor loss in the lower bound.

For $n' \ge n$, let $O^{n' \times n}$ denote the collection of $n' \times n$ orthonormal matrices $V \in \R^{n' \times n}$ i.e., $\T{V}V = I_n$. Without loss of generality, we assume that $\lambda = 1$.

Assume that there is a distribution $\calS$ over $m \times d$ matrices such that given an arbitrary matrix $A \in \R^{n \times d}$ and $b \in \R^n$ such that for $\bS \sim \calS$, with probability $\ge 99/100$,
\begin{align*}
	\opnorm{A\tilde{x} - b}^2 + \opnorm{\tilde{x}}^2
	\le (1+\varepsilon)\opt,
\end{align*}
where $\tilde{x} = \T{A}(A\T{\bS}\bS\T{A} + I)^{-1}b$. Given an instance $(A,b)$, let $\bS$ be a $\text{good}_{A,b}$ matrix if the above event holds, i.e., $\tilde{x}$ is a $1+\epsilon$ approximation. Let $b$ be a fixed unit vector. Thus, from our assumption,
\begin{equation}
	\Pr_{\bU\sim O^{n \times n}, \bV\sim O^{d \times n}, \bS \sim \calS}[\text{$\bS$ is $\text{good}_{\sigma \bU\T{\bV},b}$}] \ge 99/100.
	\label{eqn:necessary-condition}
\end{equation}
For the problem $(\sigma \bU\T{\bV}, b)$ where $b$ is a fixed unit vector, we have   $\opt = 1/(1+\sigma^2)$. We also have for $v = (\Sigma^{-2} + \T{\bV}\T{\bS}\bS \bV)^{-1}\Sigma^{-1}\T{\bU}b$ that
\begin{align*}
	\text{cost}(\tilde{x}) - \opt &= \T{v}E\Sigma(I-(\Sigma^{2} + I)^{-1})\Sigma E v\\
	&\ge \lambda_{\min}(I-(\Sigma^{2} + I)^{-1})\opnorm{\Sigma Ev}^2,
\end{align*}
where $E = \T{\bV}\T{\bS}\bS \bV - \T{\bV}\bV$, which is the error in approximating the identity matrix using the sketch $\bS$, and $\Sigma$ is the matrix of singular values of $\sigma \bU\T{\bV}$.
In our case, $\Sigma = \sigma I_n$ for some $\sigma \ge 1$ which implies that
\begin{equation*}
	\text{cost}(\tilde{x}) - \opt \ge \frac{1}{2}\opnorm{E(\sigma^{-2}I + \T{\bV}\T{\bS}\bS \bV)^{-1}\T{\bU}b}^2
\end{equation*}
once we cancel out $\Sigma$ and $\Sigma^{-1}$. Thus, if $\bS$ is $\text{good}_{(\sigma \bU\T{\bV},b)}$,
\begin{equation*}
	\opnorm{E(\sigma^{-2}I + \T{\bV}\T{\bS}\bS \bV)^{-1}\T{\bU}b}^2 \le \frac{2\varepsilon}{1+\sigma^2} \le \frac{2\varepsilon}{\sigma^2}.
\end{equation*}
Therefore,
$
	\Pr_{\bU,\bV,\bS}[\opnorm{E(\sigma^{-2}I + \T{\bV}\T{\bS}\bS \bV)^{-1}\T{\bU}b}^2\le {2\varepsilon}/{\sigma^2}]
	\ge 	\Pr_{\bU, \bV, \bS}[\text{$\bS$ is $\text{good}_{\sigma \bU\T{\bV},b}$}] \ge 99/100. 
$
Now, for a fixed unit vector $b$, the vector $\T{\bU}b$ is a uniformly random unit vector that is independent of $\bV$ and $\bS$. Thus,
\begin{equation*}
	\Pr_{\bV,\bS, \br}[\opnorm{E(\sigma^{-2}I + \T{\bV}\T{\bS}\bS \bV)^{-1}\br}^2\le {2\varepsilon}/{\sigma^2}] \ge 0.99,
\end{equation*}
where above and throughout the section, $\br$ is a uniformly random unit vector. Now we transform this property of the random matrix $\bS$ into a probability statement about the Frobenius norm of a certain matrix.

\begin{lemma}[Random vector to Frobenius Norm]
	If $\bM \in \R^{n \times n}$ is a random matrix independent of the random uniform vector $\br$ such that
$
		\Pr_{\bM, \br}[\opnorm{\bM \br}^2 \le a] \ge 99/100,
$
	then $\Pr_{\bM}[\frnorm{\bM}^2 \le Can] \ge 9/10$ for large enough constant $C$.
	\label{lma:random-vector-to-frobenius}
\end{lemma}

This lemma implies that for any random matrix $\bS$ satisfying \eqref{eqn:necessary-condition}, we have $\frnorm{E(\sigma^{-2}I + \T{\bV}\T{\bS}\bS \bV)^{-1}}^2 \le Cn\varepsilon/\sigma^2$ with probability $\ge 9/10$ over $\bV, \bS$. Using the non-dilating property of $\bS$ and applying a union bound, we now have with probability $\ge 8/10$,
\begin{align*}
	&\frnorm{E}^2 \le \frac{\frnorm{E(\sigma^{-2}I + \T{\bV}\T{\bS}\bS \bV)^{-1}}^2}{\sigma_{\min}((\sigma^{-2}I + \T{\bV}\T{\bS}\bS \bV)^{-1})^{2}}\\
	&= \frac{Cn\varepsilon/\sigma^2}{\sigma_{\min}((\sigma^{-2}I + \T{\bV}\T{\bS}\bS \bV)^{-1})^2} \le O(n\varepsilon/\sigma^2)
\end{align*}
where we used the fact that for any invertible matrix $A$, $1/\sigma_{\min}(A^{-1}) = \sigma_{\max}(A)$ and $\sigma_{\max}(\sigma^{-2}I + \T{\bV}\T{\bS}\bS \bV) \le (1/\sigma^2) + \opnorm{\T{\bV}\T{\bS}\bS\bV} = O(1)$ with probability $\ge 9/10$.
Thus, a lower bound on the number of rows in the matrix $\bS$ to obtain, with probability $\ge 8/10$,
\begin{equation}
	\frnorm{\T{\bV}\T{\bS}\bS \bV - I} \le O(\sqrt{n\varepsilon/\sigma^2}) = O(\sqrt{\varepsilon / n\sigma^2})n
	\label{eqn:amm-condition-necessary-for-regression}
\end{equation}
implies a lower bound on the number of rows of a random matrix $\bS$ that satisfies \eqref{eqn:necessary-condition}.
\subsection{Lower bounds for AMM}
\begin{lemma}
    Given parameters $n$ and error parameter $\varepsilon \le c/\sqrt{n}$ for a small enough constant $c$, for all $d \ge Cn/\varepsilon^2$, if a random matrix $\bS \in \R^{m \times d}$  for all matrices $A \in \R^{d \times n}$ satisfies,
$
        \frnorm{\T{A}\T{\bS}\bS A -\T{A}A} \le \varepsilon\frnorm{\T{A}}\frnorm{A}
$
    with probability $\ge 9/10$, then $m = \Omega(1/\varepsilon^2)$. 
    
    Moreover, the lower bound of $\Omega(1/\varepsilon^2)$ holds even for the sketching matrices that give the following guarantee:
    $
        \Pr_{\bA,\bS}[\frnorm{\T{\bA}\T{\bS}\bS \bA - I} \le \varepsilon n] \ge 0.9,
    $
    where $\bA$ is a uniformly random $d \times n$ orthonormal matrix independent of the sketch $\bS$.
    \label{lma:initital-amm-lowerbound}
\end{lemma}

Although the above lemma only shows that an AMM sketch requires $m = \Omega(1/\varepsilon^2)$ for $d \ge Cn/\varepsilon^2$, we can extend it to show the lower bound for $d \ge C/\varepsilon^2$ for a large enough constant $C$. Note that $C/\varepsilon^2 = \Omega(n)$ since $\varepsilon \le c/\sqrt{n}$.
\begin{theorem}\label{thm:final-amm-lowerbound}
	Given $n \ge 0$ and $\varepsilon < c/\sqrt{n}$ for a small enough constant $c$, there are universal constants $C,D$ such that for all $d \ge D/\varepsilon^2$, any distribution that has the $\varepsilon$ AMM property for $d \times n$ matrices must have $\ge C/\varepsilon^2$ rows.
\end{theorem}

As discussed in the introduction, we crucially use the fact that a sub-matrix of a random orthonormal matrix is close to a Gaussian matrix in total variation distance to prove the above theorem. This seems to be a useful direction to obtain lower bounds for other sketching problems.
\subsection{Lower Bound Wrapup}
In the case of ridge regression with $\lambda = 1$, \eqref{eqn:amm-condition-necessary-for-regression} shows that the sketching distribution has to satisfy the AMM guarantee with parameter $c\sqrt{\varepsilon /n\sigma^2}$. By using the above hardness result for AMM, we obtain the following theorem.
\begin{theorem}
If $\calS$ is a non-dilating distribution over $m \times d$ matrices such that for all ridge regression instances $(A,b,\lambda)$ with $A \in \R^{n \times d}$, $1 \le \sigma^2/\lambda \le \alpha$ satisfies,
\begin{equation*}
    \Pr_{\bS \sim \calS}[\opnorm{A\tilde{x} - b}^2 + \lambda\opnorm{\tilde{x}}^2 \le (1+\varepsilon)\opt] \ge 0.99,
\end{equation*}
for $\tilde{x} = \T{A}(A\T{\bS}\bS\T{A} + \lambda I)^{-1}b$, then $m = \Omega(n\alpha/\varepsilon) = \Omega(n\sigma^2/\lambda \varepsilon)$.
\end{theorem}

\section{Communication Complexity Lower Bounds for Ridge Regression}\label{sec:communication-bounds}
Consider a ridge regression matrix $A$ of the form $A_{1} + A_{2}$ where Alice has the matrix $A_1$ and Bob has the matrix $A_{2}$. To compute a vector $y$ such that $\T{A}y$ is a $1+\varepsilon$ approximation, the theorem from the previous section lower bounds the communication required between Alice and Bob by $\Omega(n\sigma^2/\lambda\varepsilon)$. The lower bound is crude in that it assumes that they only communicate $A_1\T{\bS}$ between them for some sketch $\bS$, and they compute $y$ only as $(A\T{\bS}\bS\T{A} + \lambda I)^{-1}b$.

In this section, we present communication complexity lower bounds for a different but related problem of computing a value $v$ such that $v = (1 \pm \varepsilon)\opt$ in the above two-player scheme, where Alice has the matrix $A_1$ and Bob has the matrix $A_2$. We show the lower bound by reducing from the well-known \textsc{Gap-Hamming} problem \cite{iw03,w04} to a ridge regression instance with a $2 \times d$ design matrix.

In the \textsc{Gap-Hamming} problem, Alice and Bob receive vectors $x,y \in \set{\pm 1}^d$ and want to decide if $d_H(x,y) \ge d/2 + \varepsilon d$ or $d_H(x,y) \le d/2 - \varepsilon d$, where $d_{H}(x,y)$ is the Hamming distance $|\{i\,|\, x_i \ne y_i\}|$. This problem has a communication complexity lower bound of $\Omega(1/\varepsilon^2)$ bits, even for multiple rounds \cite{chakrabarti2012optimal}. Let $M$ be a $2 \times d$ matrix with $x$ and $y$ as its rows. Consider the following ridge regression problem:
\begin{equation*}
	\min_x \opnorm{Mx - \begin{bmatrix}1 \\ -1\end{bmatrix}}^2 + \lambda \opnorm{x}^2.
\end{equation*}
Let $N = d_H(x,y)$. The optimal value of the above ridge regression problem is given by ${2\lambda}/{(\lambda + 2N)}$. In the case when $N \ge d/2 + \varepsilon d$ and $N \le d/2 + \varepsilon d$, the optimal values differ by a factor of $1 + 4\varepsilon/(1 + \lambda/d)$. So, obtaining a $1 + O(\varepsilon/(1+\lambda/d))$ approximation to ridge regression lets us solve the Gap-Hamming problem, and hence requires $\Omega(1/\varepsilon^2)$ bits of communication. As $\opnorm{M}^2 = \Theta(d)$, we obtain an $\Omega(1/\varepsilon^2(1+\lambda/\sigma^2))$ bit lower bound for computing a $1+\varepsilon$ approximate value for ridge regression.

In contrast to the $\Omega(1/\varepsilon^2)$ type communication lower bounds on computing $1 \pm \varepsilon$ approximations to the optimal values of ridge regression, we obtain $\Omega(d)$ lower bounds on the communication complexity of approximating optimal values of Lasso and square-root Lasso objectives even up to a factor of $1+c$ for a small enough constant $c > 0$. Concretely, we prove the following results.

\begin{theorem}[Communication Complexity of Lasso]
Let $0< \lambda < 1$ be the Lasso parameter. If Alice has the $n \times d$ matrix $M_1$ and Bob has the $n \times d$ matrix $M_2$, then to determine a $1+c$ approximation, for a small enough constant $c$, to the optimal value of
\begin{align*}
    \min_z \opnorm{(M_1+M_2)z - b}^2 + \lambda\|{z}\|_1, 
\end{align*}
requires $\Omega(d)$ bits of communication between Alice and Bob.
\end{theorem}

\begin{theorem}[Hardness of Sketching Square-Root Lasso]
Let $0< \lambda < 2\sqrt{2}/3$ be the Lasso parameter. If Alice has the $n \times d$ matrix $M_1$ and Bob has the $n \times d$ matrix $M_2$, then to determine a $1+c$ approximation, for a small enough constant $c > 0$, to the optimal value of
\begin{align*}
    \min_z \opnorm{(M_1+M_2)z - b} + \lambda\|{z}\|_1,
\end{align*}
requires $\Omega(d)$ bits of communication between Alice and Bob.
\end{theorem}
To show these lower bounds, we reduce the classic \textsc{Disjointness} problem \cite{haastad2007randomized} to computing $1+c$ approximations to optimal values of an appropriate Lasso and square-root Lasso problem. See the supplementary for proofs.

\section*{Acknowledgments}
The authors thank National Institute of Health (NIH) grant 5401 HG 10798-2, Office of Naval Research (ONR) grant N00014-18-1-2562, and a Simons Investigator Award.

%% file: appendix.tex
\section{Missing Proofs from Section~\ref{sec:krr}}
\begin{proof}[Proof of Lemma~\ref{lma:jl-moment-property}]
    For $i \in [m]$ and $j \in [d]$, let $\delta_{i,j}$ be the indicator random variable that denotes if the $(i,j)$-th  entry of the matrix $\bS$ is nonzero. We have that $\sum_{i}\delta_{i,j} = s$ and that for any $S \subseteq [m] \times [d]$, $\E[\Pi_{(i,j) \in S}\delta_{i,j}] \le (s/m)^{|S|}$. Also, let $\sigma_{i,j}$ be the sign of the  $(i,j)$-th entry and let $\sigma_{i,j}$ be $4$-wise independent Rademacher random variables. Now,
\begin{align*}
	\opnorm{\bS x}^2 = \sum_{i} |\bS_{i*}x|^2 = \frac{1}{s}\sum_i (\sum_{j} \delta_{i,j}\sigma_{i,j}x_j)^2 &= \frac{1}{s}\sum_i\sum_{j,j'}\delta_{i,j}\delta_{i,j'}\sigma_{i,j}\sigma_{i,j'}x_jx_{j'}\\
	&=\frac{1}{s}\sum_i\sum_j (\delta_{i,j})^2(\sigma_{i,j})^2x_j^2 + \frac{1}{s}\sum_i\sum_{j \ne j'}\delta_{i,j}\delta_{i,j'}\sigma_{i,j}\sigma_{i,j'}x_jx_{j'}.
\end{align*}
We have $\delta_{i,j}^2 = \delta_{i,j}$ and $\sigma_{i,j}^2 = 1$ for all $i,j$. So,
\begin{align*}
    (1/s)\sum_i\sum_j (\delta_{i,j})^2(\sigma_{i,j})^2x_j^2 = (1/s)\sum_{i}\sum_j \delta_{i,j}x_j^2 = (1/s)\sum_jx_j^2\sum_i\delta_{i,j} = (1/s)\sum_j x_j^2 \cdot s = \opnorm{x}^2 = 1.
\end{align*} 
Therefore,
\begin{align*}
	\opnorm{\bS x}^2 = 1 + \frac{1}{s}\sum_i\sum_{j \ne j'}\delta_{i,j}\delta_{i,j'}\sigma_{i,j}\sigma_{i,j'}x_jx_{j'}.
\end{align*}
If $\sigma_{i,j}$ are uniform random signs that are $2$-wise independent, then for $j \ne j'$, $\E[\sigma_{i,j}\sigma_{i,j'}] = 0$ and as the set of random variables $\delta_{i,j}$ are independent of the random variables $\sigma_{i,j}$, we have $\E[\delta_{i,j}\delta_{i,j'}\sigma_{i,j}\sigma_{i,j'}] = \E[\delta_{i,j}\delta_{i,j'}]\E[\sigma_{i,j}\sigma_{i,j'}] = 0$ for $j \ne j'$ which implies that $\E[\opnorm{\bS x}^2] = 1$. We also have
\begin{align*}
	(\opnorm{\bS x}^2 - 1)^2 &= \frac{1}{s^2}\sum_{i,i'}\sum_{\substack{j \ne j'\\ k\ne k'}}\delta_{i,j}\delta_{i,j'}\delta_{i',k}\delta_{i',k'}\sigma_{i,j}\sigma_{i,j'}\sigma_{i',k}\sigma_{i',k'} x_jx_{j'}x_kx_{k'}.
\end{align*}
If $i \ne i'$, $j \ne j'$, and $k \ne k'$, then the random variables $\sigma_{i,j}$, $\sigma_{i,j'}, \sigma_{i',k}$, and $\sigma_{i',k'}$ are distinct and if they are $4$-wise independent Rademacher random variables, then $\E[\sigma_{i,j}\sigma_{i,j'}\sigma_{i',k}\sigma_{i',k'}]=0$ which implies that
\begin{align*}
	\E[(\opnorm{\bS x}^2 - 1)^2] = \frac{1}{s^2}\sum_i\sum_{\substack{j \ne j'\\k \ne k'}}\E[\delta_{i,j}\delta_{i,j'}\delta_{i,k}\delta_{i,k'}]\E[\sigma_{i,j}\sigma_{i,j'}\sigma_{i,k}\sigma_{i,k'}] x_jx_{j'}x_kx_{k'}.
\end{align*}
Again, if all the indices $j,j',k,k'$ are distinct, then by the $4$-wise independence of the $\sigma$ random variables, we obtain that $\E[\sigma_{i,j}\sigma_{i,j'}\sigma_{i,k}\sigma_{i,k'}] = 0$, which leaves only $j = k \ne j' = k'$ and $j =k' \ne j' = k$ as the cases where the expectation can be non-zero. In each of these cases, $\sigma_{i,j}\sigma_{i,j'}\sigma_{i,k}\sigma_{i,k'} = 1$ with probability $1$. Therefore,
\begin{align*}
	\E[(\opnorm{\bS x}^2 - 1)^2] = \frac{2}{s^2}\sum_i\sum_{j,j'}\E[\delta_{i,j}\delta_{i,j'}]x_j^2x_{j'}^2 \le \frac{2}{s^2}\frac{s^2}{m^2}\sum_i\sum_{j \ne j'}x_j^2x_{j'}^2 \le \frac{2}{m^2} \sum_i (\sum_j x_j^2)(\sum_{j'}x_{j'}^2) \le \frac{2}{m}
\end{align*}
which gives that $\|\opnorm{\bS x}^2 - 1\|_{L^2} = \E[(\opnorm{\bS x}^2 -1)^2]^{1/2} \le \sqrt{2/m}$. Now, for $m = \Omega(1/\varepsilon^2\delta)$, we have $\|\opnorm{\bS x}^2 - 1\|_{L^2} \le \varepsilon\delta^{1/2}$, which proves that the matrix $\bS$ has the $(\varepsilon, \delta, 2)$-JL moment property.
\end{proof}

\section{Missing Proofs from Section~\ref{sec:lowerbounds}}
\subsection{Proof of Lemma~\ref{lma:random-vector-to-frobenius}}
Lemma~\ref{lma:random-vector-to-frobenius} transforms a probability statement about the squared norm of a product of a random matrix $\bM$ with an independent uniform random vector $\br$. To prove the lemma, we first prove the following similar lemma in which the matrix $M$ is a deterministic matrix.
\begin{lemma}
    Let $M \in \R^{n \times n}$ be a fixed matrix and $\br$ be a uniformly random unit vector. If $\Pr_{\br}[\opnorm{M\br}^2 \le a] \ge 9/10$, then $\frnorm{M}^2 \le Cna$ for a large enough universal constant $C$.
\end{lemma}
\begin{proof}
    Let $\bg \in \R^n$ be a Gaussian random vector with i.i.d. entries drawn from $N(0,1)$. Then the distribution of $\bg/\opnorm{\bg}$ is identical to that of a uniformly random unit vector in $n$ dimensions by rotational invariance of the Gaussian distribution. Therefore from our assumption, $\Pr_\bg[\opnorm{M\bg}^2 \le a\opnorm{\bg}^2] \ge 9/10$. We also have that with probability $\ge 9/10$, $\opnorm{\bg}^2 \le C_1n$ for a large enough absolute constant $C_1$. Thus, we have by a union bound that,
\begin{equation*}
	\Pr_\bg[\opnorm{M\bg}^2 \le a\opnorm{\bg}^2\ \land \ \opnorm{\bg}^2 \le C_1n] \ge 8/10,
\end{equation*}
which implies that
\begin{equation*}
	\Pr_{\bg}[\opnorm{M\bg}^2 \le C_1an] \ge 8/10.
\end{equation*}
Let $M = U\Sigma \T{V}$ be the singular value decomposition of the matrix $M$. Then, the above equation is equivalent to 
\begin{equation*}
	8/10 \le \Pr_{\bg}[\opnorm{\Sigma \T{V}\bg}^2 \le C_1an] = \Pr_\bg[\opnorm{\Sigma \bg}^2 \le C_1an]
\end{equation*}
where the equality follows from the fact that for an orthonormal matrix $\T{V}$, we have $\T{V}\bg \equiv \bg$. Thus, if the singular values of $M$ are $\sigma_1,\ldots, \sigma_n$, we have
\begin{equation*}
	\Pr_{\bg}[\sum_i \sigma_i^2\bg_i^2 \le C_1an] \ge 8/10.
\end{equation*}
Now, we have the following lemma which gives an upper bound on the probability of a linear combination of squared Gaussian random variables being small.
\begin{lemma}[\citet{anti-concentration}]
\label{lma:anti-concentration}
	If $a_i \ge 0$ for $i=1,\ldots, n$ are constants and $\bg_1,\ldots, \bg_n$ are i.i.d. Gaussians of mean $0$ and variance $1$, then for every $\delta > 0$,
	\begin{equation*}
		\Pr[\sum_i a_i \bg_i^2 \le \delta \sum_i a_i] \le e\sqrt{\delta}.
	\end{equation*}
\end{lemma}
\begin{proof}
The inequality is obviously true for $\delta\ge 1$. We now consider arbitrary $\delta < 1$. Assume without loss of generality that $\sum_i a_i$ = 1. Now, for any $\lambda > 0$,
\begin{equation*}
\Pr[\sum_i a_i \bg_i^2 \le \delta] = \Pr[-\lambda \sum_i a_i \bg_i^2 \ge -\lambda \delta] = \Pr[\exp(-\lambda\sum_i a_i \bg_i^2) \ge \exp(-\lambda \delta)] \le \exp(\lambda \delta)\E[\exp(-\lambda \sum_i a_i\bg_i^2)]
\end{equation*}	
and therefore,
\begin{align*}
	\Pr[\sum_i a_i \bg_i^2 \le \delta] &= \exp(\lambda \delta)\E[\exp(-\lambda \sum_i a_i\bg_i^2)]\\
	&= \exp(\lambda \delta)\prod_i\E[\exp(-\lambda a_i\bg_i^2)]\\
	&= \exp(\lambda \delta)\prod_i (1+2\lambda a_i)^{-1/2}.
\end{align*}
Now, $\prod_i (1+2\lambda a_i) \ge 1 + 2\lambda(\sum_i a_i) = 1 + 2\lambda$ which implies that $\prod (1+2\lambda a_i)^{-1/2} \le (1+2\lambda)^{-1/2}$ which gives
\begin{align*}
	\Pr[\sum_i a_i \bg_i^2 \le \delta] \le \exp(\lambda \delta)(1+2\lambda)^{-1/2}.
\end{align*}
Picking $\lambda \ge 0$ such that $1 + 2\lambda = 1/\delta$, we obtain that $\Pr[\sum_i a_i \bg_i^2 \le \delta] \le \exp((1-\delta)/2)\sqrt{\delta} \le e\sqrt{\delta}$.
\end{proof}

For $\delta = 0.01$, the above lemma implies that $\Pr[\sum_i \sigma_i^2 \bg_i^2 \le 0.01\sum_i \sigma_i^2] \le e \cdot (0.1) \le 0.3$. This, in particular implies that $0.01\sum_i \sigma_i^2 = 0.01 \frnorm{\Sigma}^2 \le C_1an$ which gives
	$\frnorm{M}^2 = \frnorm{\Sigma}^2 \le Can$ for a large enough absolute constant $C$.
\end{proof}
We now extend the above lemma to the case when the matrix $\bM$ is also random and independent of the random unit vector $\br$.
\begin{proof}[Proof of Lemma~\ref{lma:random-vector-to-frobenius}]
	Let $\bM$ be \emph{good} if $\Pr_{\br}[\opnorm{\bM \br}^2 \le a] \ge 9/10$ and let $\bM$ be \emph{bad} otherwise and note from the above lemma that if $\bM$ is \emph{good}, then $\frnorm{\bM}^2 \le Can$. Now,
	\begin{align*}
		&99/100 \le \Pr_{\bM, \br}[\opnorm{\bM \br}^2 \le a]\\
		&\le \Pr_{\bM}[\text{$\bM$ is \emph{good}}] + \Pr_{\bM}[\text{$\bM$ is \emph{bad}}] \cdot (9/10)\\
		&= 9/10 + (1/10) \cdot \Pr_{\bM}[\text{$\bM$ is \emph{good}}]
	\end{align*}
	which implies that $\Pr_{\bM}[\text{$\bM$ is \emph{good}}] \ge 9/10$ and therefore $\Pr_{\bM}[\frnorm{\bM}^2 \le Can] \ge \Pr_{\bM}[\text{$\bM$ is good}] \ge 9/10$.
\end{proof}
\subsection{Proof of Lower Bounds for AMM}
\begin{proof}[Proof of Lemma~\ref{lma:initital-amm-lowerbound}]
We assume that such a distribution exists with $m \le c/\varepsilon^2$ for a small enough constant $c$. Let $\bA \in \mathbb{R}^{d \times n}$ be a uniformly random orthonormal matrix ($\T{\bA}\bA = I_n$) independent of the sketching matrix. Then we have
\begin{equation*}
	\Pr_{\bA,\bS}[\frnorm{\T{\bA}\T{\bS}\bS \bA - I} \le \varepsilon n] \ge 0.9,
\end{equation*}
as $\frnorm{\T{A}} = \sqrt{n}$. Let $\bS = U\Sigma \T{V}$ be the singular value decomposition with $U \in \mathbb{R}^{m \times m}$, $\Sigma \in \mathbb{R}^{m \times m}$ and $\T{V} \in \R^{m \times d}$. Note that if $\bS$ is a random matrix that satisfies the AMM property, then $\bS \cdot \bQ$ is also a random matrix that satisfies the AMM property where $\bQ$ is an independent uniformly random orthonormal matrix. Therefore, we can without loss of generality assume that $\Sigma$ is independent of $\T{V}$ and that $\T{V}$ is a uniformly random orthonormal matrix. Thus, the above condition implies that
\begin{equation*}
	\Pr_{\bA, V,\Sigma}[\frnorm{\T{\bA}V\Sigma^2\T{V}\bA - I} \le \varepsilon n] \ge 0.9.
\end{equation*}
Using the following lemma, we effectively show that the matrix $\T{V}\bA$ in the above statement can be replaced with $(1/\sqrt{d})\hat{\bG}$, where $\hat{\bG}$ is a Gaussian matrix of the same dimensions as $\T{V}\bA$.     
\begin{lemma}[Lemma~3 of \citet{product-of-gaussians}]
Let $\bG \sim \mathcal{G}_{d,d}$ and $\bZ \sim O^{d \times d}$. Suppose that $p,q \le d$ and $\hat{\bG}$ is the top left $p \times q$ block of $\bG$ and $\hat{Z}$ is the top left $p \times q$ block of $Z$. Then 
$
    d_{\text{KL}}(\frac{1}{\sqrt{d}}\hat{\bG} \| \hat{Z}) \le C\frac{pq}{d},
$
where $C$ is a universal constant. By applying Pinsker's inequality, we obtain that 
\begin{equation*}
d_{\text{TV}}(\frac{1}{\sqrt{d}}\hat{\bG} \| \hat{Z}) \le \sqrt{(1/2)d_{\text{KL}}(\frac{1}{\sqrt{d}}\hat{\bG}\|\hat{Z})}\le \sqrt{Cpq/2d}.	
\end{equation*}
\end{lemma}
Now both the matrices $V,\bA$ can be taken to be the first $m$ and $n$ columns of independent uniform random orthogonal matrices $\bV'$ and $\bA'$, respectively. By the properties of the Haar Measure, we obtain that $\T{\bV'}\bA'$ is also a uniform random orthogonal matrix. Thus, the matrix $\T{V}\bA$ can be seen as the top left $m \times n$ sub-matrix of a uniformly random orthogonal matrix. If $nm \le  d/100C$, which can be assumed as $m \le c/\varepsilon^2$ for a small enough constant $\alpha$,  we have from the above lemma that $d_{\text{TV}}(\frac{1}{\sqrt{d}}\bG \| \T{V}\bA) \le 0.1$ which implies that
\begin{align*}
	|\Pr_{\bG, \Sigma}[\frnorm{(1/d)\T{\bG}\Sigma^2\bG - I} \le \varepsilon n] - \Pr_{\bA, V,\Sigma}[\frnorm{\T{\bA}V\Sigma^2\T{V}\bA - I} \le \varepsilon n]| \le 0.1
\end{align*}
and therefore
\begin{equation}
    \Pr_{\bG, \Sigma}[\frnorm{(1/d)\T{\bG}\Sigma^2\bG - I} \le \varepsilon n] \ge 0.8,
    \label{eqn:condition-on-sigma-g}
\end{equation}
where $\bG$ is an $m \times n$ matrix of i.i.d. normal random variables. We will now show that if $m \ll 1/\varepsilon^2$, then no distribution over matrices $\Sigma$ satisfies the above condition. Note that $\bG$ and $\Sigma$ are independent. We prove this by showing that a random matrix $\Sigma$ satisfying the above probability statement must satisfy two properties simultaneously that cannot be satisfied unless $m \ge c/\varepsilon^2$ for a large enough constant $c$.

Let $\bG_{l}$ denote the left half of the matrix $\bG$ and $\bG_r$ denote the right half of the matrix $\bG$. We have
\begin{equation*}
	\frnorm{(1/d)\T{\bG}\Sigma^2\bG - I}^2 \ge \frac{1}{d^2}\frnorm{\T{\bG_r}\Sigma^2\bG_l}^2
\end{equation*}
which is obtained by considering the Frobenius norm of only the bottom-left block matrix of $(1/d)\T{\bG}\Sigma^2\bG - I$.
We first have the following lemma.
\begin{lemma}
	Let $M$ be a fixed matrix and $\bG$ be a Gaussian matrix with $t$ columns. Then with probability $\ge 0.9$, $\frnorm{M\bG}^2 \ge 0.001t\frnorm{M}^2$.
\end{lemma}
\begin{proof}
	Let $M = U\Sigma\T{V}$. We have $M\bG = U\Sigma\T{V}\bG = U\Sigma \bG'$ where $\bG'$ is a Gaussian matrix with $t$ columns. Now, $\frnorm{M\bG}^2 = \frnorm{U\Sigma \bG'}^2 = \frnorm{\Sigma \bG'}^2 = \sum_i\sum_j \sigma_i^2g_{ij}^2$. By Lemma~\ref{lma:anti-concentration}, $ \sum_i\sum_j \sigma_i^2g_{ij}^2 \ge (\sum_i \sum_j \sigma_i^2) \cdot 0.001$ with probability $\ge 0.9$. Now, using the equality $\sum_i \sum_j \sigma_i^2 = \sum_i t\sigma_i^2 = t\frnorm{\Sigma}^2 = t\frnorm{M}^2$, we finish the proof. 
\end{proof}
Thus, conditioned on the matrix $\T{\bG_r}\Sigma^2$, we have that with probability $\ge 0.9$, \[\frnorm{\T{\bG_r}\Sigma^2 \bG_l}^2 \ge 0.001(n/2)\frnorm{\T{\bG_r}\Sigma^2}^2.\] Applying the same lemma again, we have with probability $\ge 0.9$, $\frnorm{\T{\bG_r}\Sigma^2}^2 \ge 0.001(n/2)\frnorm{\Sigma^2}^2$. Thus, overall with probability $\ge 0.8$ over $\bG$, we have for any fixed $\Sigma$ that,
$
	\frnorm{\T{\bG_r}\Sigma^2 \bG_l}^2 \ge \Omega(n^2\frnorm{\Sigma^2}^2).
$ and therefore,
\begin{equation*}
	\Pr_{\bG, \Sigma}[\frnorm{\T{\bG_r}\Sigma^2\bG_{l}}^2 \ge \Omega(n^2\frnorm{\Sigma^2}^2)] \ge 0.8.
\end{equation*}
Using a union bound with \eqref{eqn:condition-on-sigma-g}, we obtain that with probability $\ge 0.6$, it is simultaneously true that
\begin{equation*}
\varepsilon^2n^2 \ge \frnorm{(1/d)\T{\bG}\Sigma^2\bG - I}^2 \ge \frac{1}{d^2}\frnorm{\T{\bG_r}\Sigma^2\bG_l}^2
\end{equation*}
and
\begin{equation*}
    \frnorm{\T{\bG_r}\Sigma^2\bG_{l}}^2 \ge \Omega(n^2\frnorm{\Sigma^2}^2)
\end{equation*}
which implies that with probability $\ge 0.6$, $(1/d^2)\frnorm{\Sigma^2}^2 = O(\varepsilon^2)$ i.e., $(1/d^2)\sum_{i=1}^m \sigma_i^4 = O(\varepsilon^2)$. Thus, if $\bS$ is a random matrix that satisfies the AMM property and if $\sigma_1,\ldots, \sigma_r$ are its singular values, then with probability $\ge 0.6$,
\begin{equation}
	\sum_i \sigma_i^4 \le C_1d^2\varepsilon^2
	\label{eqn:first-implication-sigma}
\end{equation}
for a universal constant $C_1$.

We now obtain a different probability statement about the singular values of the sketching matrix $\bS$ by considering the sum of squares of the diagonal entries of the matrix $(1/d)\T{\bG}\Sigma^2\bG - I = (1/d)\sum_{i=1}^m \sigma_i^2\bg_i\T{\bg_i} - I$ where $\bg_i$ are $n$ dimensional Gaussian vectors. Note that $((1/d)\T{\bG}\Sigma^2\bG - I)_{jj} = (1/d)\sum_{i=1}^m \sigma_i^2\bg_{ij}^2 - 1$. Fix the matrix $\Sigma$. Clearly,
\begin{equation*}
	\frnorm{(1/d)\T{\bG}\Sigma^2 \bG - I}^2 \ge \sum_{j=1}^n ((1/d)\sum_{i=1}^m \sigma_i^2\bg_{ij}^2 - 1)^2.
\end{equation*}
If $\sum_{i=1}^m\sigma_i^2 \le d/100$, we have that with probability at least $0.9$, $(1/d)\sum_{i=1}^m \sigma_i^2 \bg_{ij}^2 \le (10/d)\E[\sum_{i=1}^m \sigma_i^2\bg_{ij}^2] \le (10/d) \cdot (d/100)$ which implies that $((1/d)\sum_{i=1}^m \sigma_i^2\bg_{ij}^2 - 1)^2 \ge 1/4$ with probability $\ge 0.9$. Let $j \in [n]$ be \emph{large} if the previous event holds. By a Chernoff bound, with probability $\ge 0.9$, there are $\ge n/C_2$ \emph{large} values $j \in [n]$ for a large enough absolute constant $C_2$. 
Thus, $\sum_{i=1}^m \sigma_i^2  \le d/100$ implies that with probability $\ge 0.9$, $\frnorm{(1/d)\T{G}\Sigma^2 G - I}^2 \ge (n/C_2)(1/4) = n/4C_2 \ge \varepsilon^2 n^2$ as we assumed that $\varepsilon \le c/\sqrt{n}$ for a small enough constant $c$. Now, if $\Pr_{\Sigma}[\sum_{i=1}^m \sigma_i^2 \le d/100] > 0.3$, then by the above property for a fixed $\Sigma$, $\Pr_{\Sigma ,G}[\frnorm{(1/d)\T{\bG}\Sigma^2 \bG - I}^2 \ge \varepsilon^2n^2] > 0.2$ which implies that \[\Pr_{\Sigma, \bG}[\frnorm{(1/d)\T{\bG}\Sigma^2 \bG - I}^2 \le \varepsilon^2 n^2] < 0.8\] which is a contradiction to \eqref{eqn:condition-on-sigma-g}. Thus, 	$\Pr_{\Sigma}[\frnorm{\Sigma}^2 \le d/100] < 0.3$ which implies 
\begin{equation}
\Pr_{\Sigma }[\sum_{i=1}^m \sigma_i^2 \ge d/100] \ge 0.7.
	\label{eqn:second-implication-sigma}
\end{equation}
By a union bound on \eqref{eqn:first-implication-sigma} and \eqref{eqn:second-implication-sigma}, with probability $\ge 0.3$, it simultaneously holds that
\begin{equation*}
	\sum_{i=1}^m \sigma_i^4 \le C_1d^2\varepsilon^2 \ \text{and}\ \sum_{i=1}^m \sigma_i^2 \ge d/100.
\end{equation*}
Now,
\begin{align*}
	d^2/100^2 \le \left(\sum_{i=1}^m \sigma_i^2\right)^2 \le m\sum_{i=1}^m \sigma_i^4 \le C_1md^2\varepsilon^2.
\end{align*}
Here we used the Cauchy-Schwarz inequality which finally implies that $m = \Omega(1/\varepsilon^2)$. Thus, any oblivious distribution that gives AMM with $\varepsilon < c/\sqrt{n}$ for a small enough constant $c$ must have $\Omega(1/\varepsilon^2)$ rows.
\end{proof}

Before proving Theorem~\ref{thm:final-amm-lowerbound}, we first prove the following lemma that shows CountSketch preserves the Frobenius norm of a matrix.
\begin{lemma}[CountSketch Preserves Frobenius Norms]
If $\bS$ is a CountSketch matrix with $m \ge 200/\varepsilon^2$, then for any arbitrary matrix $A$, with probability $\ge 9/10$,
\begin{equation*}
	\frnorm{\bS A}^2 = (1 \pm \varepsilon)\frnorm{A}^2.
\end{equation*}
\label{lma:countsketch-fro-norm}
\end{lemma}
\begin{proof}
	For any vector $x$, we have $\E[(\opnorm{\bS x}^2 - \opnorm{x}^2)^2] \le (2/m)\opnorm{x}^4$ if $\bS$ is a CountSketch matrix with $m$ rows. Now,
	\begin{align*}
		\E[|\opnorm{\bS x}^2 - \opnorm{x}^2|]^2 \le \E[(\opnorm{\bS x}^2 - \opnorm{x}^2)^2] \le (2/m)\opnorm{x}^4
	\end{align*}
	which implies that $\E[|\opnorm{\bS x}^2 - \opnorm{x}^2|] \le \sqrt{2/m}\opnorm{x}^2$. For any arbitrary matrix $A$, we have $|\opnorm{\bS A}^2 - \opnorm{A}^2| = |\sum_i \frnorm{\bS A_{*i}}^2 - \frnorm{A_{*i}}^2| \le \sum_i |\opnorm{\bS A_{*i}}^2 - \opnorm{A_{*i}}^2|$. Thus, $\E[|\frnorm{\bS A}^2 - \frnorm{A}^2|] \le \sum_i \E[|\opnorm{\bS A_{*i}}^2 - \opnorm{A_{*i}}^2|] \le \sqrt{2/m}\sum_i\opnorm{A_{*i}}^2 = \sqrt{2/m}\frnorm{A}^2$. For $m \ge 200/\varepsilon^2$, we have $\E[|\frnorm{\bS A}^2 - \frnorm{A}^2|] \le (\varepsilon/10)\frnorm{A}^2$. By Markov's inequality, with probability $\ge 9/10$, $\frnorm{\bS A}^2 = (1 \pm \varepsilon)\frnorm{A}^2$.
\end{proof}

\begin{proof}[Proof of Theorem~\ref{thm:final-amm-lowerbound}]
Given $n$ and $\varepsilon \le c/\sqrt{n}$ for a small enough constant, assume that for $d = C_1/\varepsilon^2$ for a large enough universal constant $C_1$, there is a random matrix $\bS$ with $r < C_2/\varepsilon^2$ rows such that for any fixed matrix $A \in \R^{d \times n}$, with probability $\ge 99/100$,
\begin{equation*}
	\frnorm{\T{A}\T{\bS}\bS A - \T{A}A} \le (\varepsilon/3)\frnorm{\T{A}}\frnorm{A}.
\end{equation*}
Now consider an arbitrary matrix $B \in \R^{d' \times n}$ for $d' \ge Cn/\varepsilon^2$ for which the previous lemma applies. Let $\bS_1$ be a CountSketch matrix with $K/\varepsilon^2$ rows for a large enough $K$. With probability $\ge 95/100$, we simultaneously have (i) $\frnorm{\T{B}\T{\bS_1}\bS_1 B - \T{B}B} \le (\varepsilon/3)\frnorm{\T{B}}\frnorm{B} = (\varepsilon/3)\frnorm{B}^2$, and
	(ii) $\frnorm{\bS_1 B} = (1 \pm \varepsilon/3)\frnorm{B}$.
By picking $C_1$ large enough, we have that $C_1 \ge K$. Thus, by our assumption, the random matrix $\bS$ gives the AMM property for the matrix $\bS_1 A$. Conditioning on the above events, we have with probability $\ge 99/100$ that
\begin{align*}
	&\frnorm{\T{B}\T{\bS_1}\T{\bS}\bS\bS_1 B - \T{B}\T{\bS_1}\bS_1 B}\\
	&\le (\varepsilon/3)\frnorm{\T{B}\T{\bS_1}}\frnorm{\bS_1 B}\\
	&\le (\varepsilon/3)(1+\varepsilon/3)^2\frnorm{B}^2 \le (\varepsilon/2)\frnorm{B}^2.
\end{align*}
By the triangle inequality, we obtain $\frnorm{\T{B}\T{\bS_1}\T{\bS}\bS\bS_1 B - \T{B}B} \le (\epsilon/3 + \epsilon/2)\frnorm{B}^2 \le \varepsilon\frnorm{B}^2$. Thus, by a union bound, with probability $\ge 0.9$, the random matrix $\bS \cdot \bS_1$ satisfies that for any fixed matrix $B$ with $\ge Cn/\varepsilon^2$ rows, with probability $\ge 0.9$,
\begin{equation*}
	\frnorm{\T{B}\T{(\bS \cdot \bS_1)}(\bS \cdot \bS_1)B - \T{B}B} \le \varepsilon\frnorm{B}^2
\end{equation*}
implying that even for a matrix with at least $ Cn/\varepsilon^2$ rows, there is an oblivious sketching distribution with $r < C_2/\varepsilon^2$ rows which gives an AMM guarantee. This contradicts the previous lemma and hence our assumption that there is a small sketching matrix for matrices with $d = C_1/\varepsilon^2$ rows is false. Thus, we have that for any $d \ge C/\varepsilon^2$ for a large enough constant $C$, there is no sketching distribution with $< C_1/\varepsilon^2$ rows that gives the $\varepsilon$ AMM guarantee for matrices with $\ge d$ rows.
\end{proof}
\section{Missing Proofs from Section~\ref{sec:communication-bounds}}
\subsection{Lower Bounds for Ridge Regression}
Note that multiplying a column of $M$ with $-1$ does not change the optimum value. Therefore, without loss of generality, we can assume that the columns of $M$ are either $\T{[1\, 1]}$ or $\T{[1\, {-1}]}$. Thus, we have that $d-N$ columns of $M$ are equal to $\T{[1\, 1]}$ and $N$ columns of the matrix $M$ are equal to $\T{[1\, {-1}]}$. Let $I \subseteq [n]$, $|I |=N$ be the set of columns of $M$ that are equal to $\T{[1\, {-1}]}$. Now, for any vector $x$,
\begin{align*}
	&\opnorm{Mx - \begin{bmatrix}1 \\ -1\end{bmatrix}}^2 + \lambda \opnorm{x}^2\\
	&= (\sum_{i \in I} x_i + \sum_{i \notin I}x_i - 1)^2 + (-\sum_{i \in I} x_i+ \sum_{i \notin I}x_i + 1)^2\\
	&\quad  + \lambda \sum_{i \in I} x_i^2 + \lambda\sum_{i \notin I} x_i^2\\
	&= 2(\sum_{i \notin I} x_i)^2 + 2(\sum_{i \in I} x_i -1)^2 + \lambda \sum_{i \in I}x_i^2 + \lambda\sum_{i \notin I}x_i^2.
\end{align*}
Clearly, to minimize the expression, we need to set $x_i = 0$ for $i \notin I$. We also have that for a fixed value of $\sum_{i \in I} x_i$, the expression $\sum_{i \in I}x_i^2$ is minimized when $x_i = x_{i'}$ for all $i,i' \in I$. Let $x_i = \alpha$ for all $i \in I$. Then, 
\begin{align*}
	\opnorm{Mx - \begin{bmatrix}1 \\ -1\end{bmatrix}}^2 + \lambda \opnorm{x}^2 &= 2(N\alpha -1)^2 + \lambda N\alpha^2.
\end{align*}
To minimize the expression, we set $\alpha = 2/(2N + \lambda)$ and obtain that
\begin{equation*}
	\min_x \opnorm{Mx - \begin{bmatrix}1 \\ -1\end{bmatrix}}^2 + \lambda \opnorm{x}^2 = \frac{2\lambda}{\lambda + 2N}.
\end{equation*}
\section{Communication Lower Bounds for other Regularized Problems}
We show communication lower bounds for computing $1+\varepsilon$ approximations to optimum values of Lasso and the square-root Lasso problems by reducing the \textsc{Disjointness} problem to computing these optimum values. 

In an instance of \textsc{Disjointness} problem Alice receives a set $A \subseteq [n]$ and Bob receives a set $B \subseteq [n]$. They send bits to each other to communicate based on a pre-determined to protocol to find if $A \cap B = \emptyset$ or not.
\begin{theorem}[Hardness of \textsc{Disjointness} \cite{haastad2007randomized}]
	Every public-coin randomized protocol for \textsc{Disjointness} that has two-sided error at most a constant $\varepsilon \in (0, 1/2)$ uses $\Omega(n)$ communication in the worst case (over inputs and coin flips).
\end{theorem}
\begin{theorem}[Communication Complexity of Lasso]
Let $0< \lambda < 1$ be the Lasso parameter. If Alice has the $n \times d$ matrix $M_1$ and Bob has the $n \times d$ matrix $M_2$, then to determine a $1+c$ approximation, for a small enough constant $c$, to the optimal value of
\begin{align*}
    \min_z \opnorm{(M_1+M_2)z - b}^2 + \lambda\|{z}\|_1
\end{align*}
requires $\Omega(d)$ bits of communication between Alice and Bob.
\end{theorem}
\begin{proof}
We prove the theorem by reducing the \textsc{Disjointness} problem to an instance of the Lasso problem. Let $A, B \subseteq [d]$ denote the sets received by Alice and Bob, respectively. Let $x^{(1)}, x^{(2)} \in \{0,1\}^d$ be the vectors that denote $A$ and $B$ respectively. We have $x^{(1)}_i = 1$ if and only if $i \in A$ and similarly $x^{(2)}_i = 1$ if and only if $i \in B$. Let $M$ denote a $2 \times d$ matrix with two rows such that $M_{1*} = x^{(1)}$ and $M_{2*}= x^{(2)}$. We consider the Lasso problem:
\begin{equation*}
	\min_z \opnorm{Mz - \begin{bmatrix}
		1 \\ 1
	\end{bmatrix}}^2 + \lambda \onenorm{z}
\end{equation*}
for some $0 < \lambda < 1$.

Consider the following two cases:
\begin{enumerate}
	\item $A \cap B \ne \emptyset$ : Let $i \in A \cap B$. We have $x^{(1)}_i = x^{(2)}_i = 1$ and therefore $M_{1i} = M_{2i} = 1$. Let $z = e_i$. We have 
	$Me_i = \T{[1\quad 1]}$. Therefore
	\begin{equation*}
		\opnorm{Me_i - \begin{bmatrix}
		1 \\ 1
	\end{bmatrix}}^2 + \lambda\onenorm{e_i} = 0 + \lambda = \lambda.
	\end{equation*}
	Therefore $\min_z \opnorm{Mz - 1_2}^2 + \lambda\onenorm{z} \le \lambda$.
	\item Let $A \cap B = \emptyset$. Let $z$ be an arbitrary vector. We have $M_{1*}z = \sum_{i \in A}z_i$ and $M_{2*}z = \sum_{i \in B}z_i$. Thus,
	\begin{align*}
		\opnorm{Mz - 1_2}^2 + \lambda\onenorm{z} = (\sum_{i \in A}z_i - 1)^2 +(\sum_{i \in B}z_i -1)^2 + \lambda (\sum_{i \in A}{|z_i|} + \sum_{i \in B}{|z_i|}),
	\end{align*}
We can see that the optimum value can be attained by setting $z_i = \alpha$ for some $i \in A$ and $z_j = \beta$ for some $j \in B$ and setting the rest of the coordinates to $0$. We now consider
\begin{equation*}
	\min_{\alpha} (\alpha - 1)^2 + \lambda\alpha.
\end{equation*}
The optimal solution for the above is attained at $\alpha$ satisfying $2(\alpha - 1) + \lambda = 0$ which gives $\alpha = 1 - \lambda/2$. Thus the optimum value for 
\begin{equation*}
	\min_z \opnorm{Mz - 1_2}^2 + \lambda\onenorm{z}
\end{equation*}
when $A \cap B = \emptyset$ is given by 
\begin{equation*}
	2 \cdot \frac{\lambda^2}{4} + \lambda(2 - \lambda) = 2\lambda -  \lambda^2/2 \ge 3\lambda/2.
\end{equation*}
\end{enumerate}
The optimal Lasso values differ by a factor of at least $3/2$. Thus a protocol that can compute $1+c$ approximation to the optimal value of ridge regression can be used to solve the \textsc{Disjointness} problem which gives an $\Omega(d)$ bits lower bound for approximating Lasso.
\end{proof}

We similarly have the following communication lower bounds for Square-root Lasso.
\begin{theorem}[Hardness of Sketching Square-Root Lasso]
Let $0< \lambda < 2\sqrt{2}/3$ be the Lasso parameter. If Alice has the $n \times d$ matrix $M_1$ and Bob has the $n \times d$ matrix $M_2$, then to determine a $1+c$ approximation, for a small enough constant $c$, to the optimal value of
\begin{align*}
    \min_z \opnorm{(M_1+M_2)z - b} + \lambda\|{z}\|_1,
\end{align*}
requires $\Omega(d)$ bits of communication between Alice and Bob.
\end{theorem}
\begin{proof}
	We use the same notation as in the proof of the above theorem. In the case of $A \cap B \ne \emptyset$, we again have that the optimum value is at most $\lambda$. In the other case, where $A \cap B = \emptyset$, we again have that the optimum value can be attained by setting $z_i = \alpha = z_j$ for some $i \in A$ and $j \in B$. Thus we have the following optimization problem
\begin{equation*}
	\min_{\alpha} \sqrt{2}|\alpha -1| + 2\lambda|\alpha| = \min(\sqrt{2}, 2\lambda).
\end{equation*}
Thus, the optimum value of the square-root Lasso problem is at most $\lambda$ in the case $A \cap B \ne \emptyset$ and is at least $\min(\sqrt{2}, 2\lambda)$ in the case of $A \cap B = \emptyset$. If $\lambda < 2\sqrt{2}/3$, then $\min(\sqrt{2}, 2\lambda) \ge (3\lambda/2)$. Thus, if Alice and Bob can obtain an $11/10$ approximation for the Square-Root Lasso instance $(M, 1_2, \lambda)$ for some $\lambda < 2\sqrt{2}/3$, they can solve the \textsc{Disjointness} instance which implies an $\Omega(d)$ bit lower bound for computing a $1+c$ approximation for Square-Root Lasso.
\end{proof}
\section{An Experiment}
We run our algorithm on a ridge regression instance with a $6000 \times 70000$ matrix $A$ whose entries are independent Gaussian random variables. We set $\lambda$ such that $\sigma^2/\lambda \approx 1$. Na\"ively computing $x^* = \T{A}(A\T{A} + \lambda I)^{-1}b$ takes $t_{\text{naive}} = 71$ seconds on our machine. We use OSNAP with sparsity $s = 8$ and vary the number $r$ of rows and observe the running times and quality of the solution that is obtained by our algorithm.

Our experiments show the general trends we expect. Increasing the number of rows in the sketching matrix results in a solution $\tilde{x}$ that has lower cost and also is closer to the optimum solution $x^*$. We also see that the running time of the algorithm is nearly linear in the sketch size $r$, implying that the time required to apply the sketch is negligible for this instance. At sketch size $r = 30000$, that is less than $d/2$, we see that the algorithm runs nearly $40\%$ faster than the na\"ive algorithm while computing a solution that has a cost within $5\%$ of the optimum. For  larger values of  $d$, we expect to  obtain a greater  speedup as  compared  to na\"ively computing $x^*$.

Notice that we do not compare with the algorithm of \citet{chowdhury2018iterative} as for one iteration, our algorithm is exactly the same as theirs. Our theorems show that the sketch can be smaller and sparser than what is shown in their work to compute $1+\varepsilon$ approximate solutions, giving the first proof of correctness about the quality of the solution at smaller sketch sizes.
\begin{figure}[!htb]
\minipage{0.4\textwidth}
  \includegraphics[width=\linewidth]{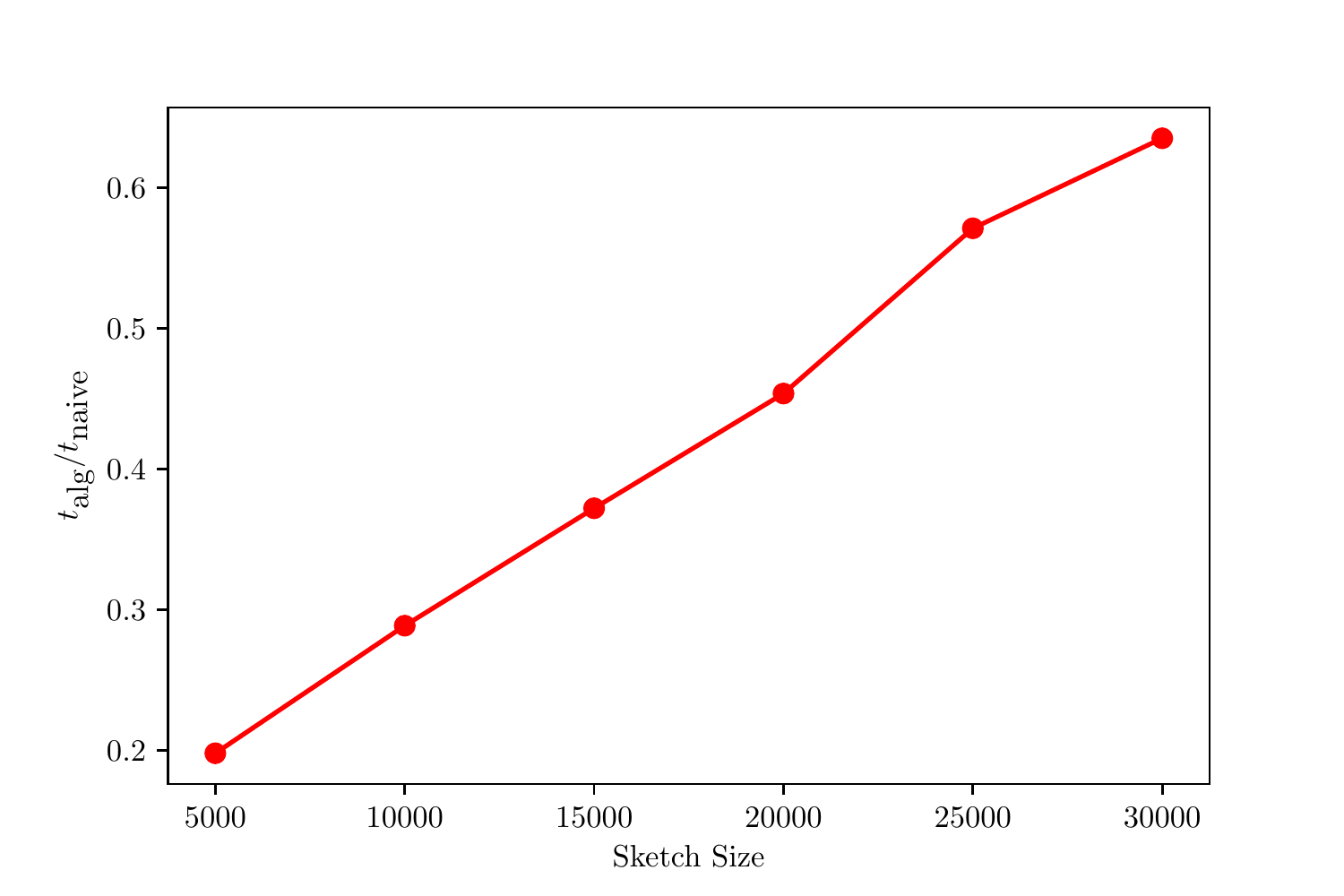}
  \caption{$t_{\text{alg}}/t_{\text{naive}}$ vs \# of rows of OSNAP}\label{fig:awesome_image1}
\endminipage\hfill
\minipage{0.4\textwidth}
  \includegraphics[width=\linewidth]{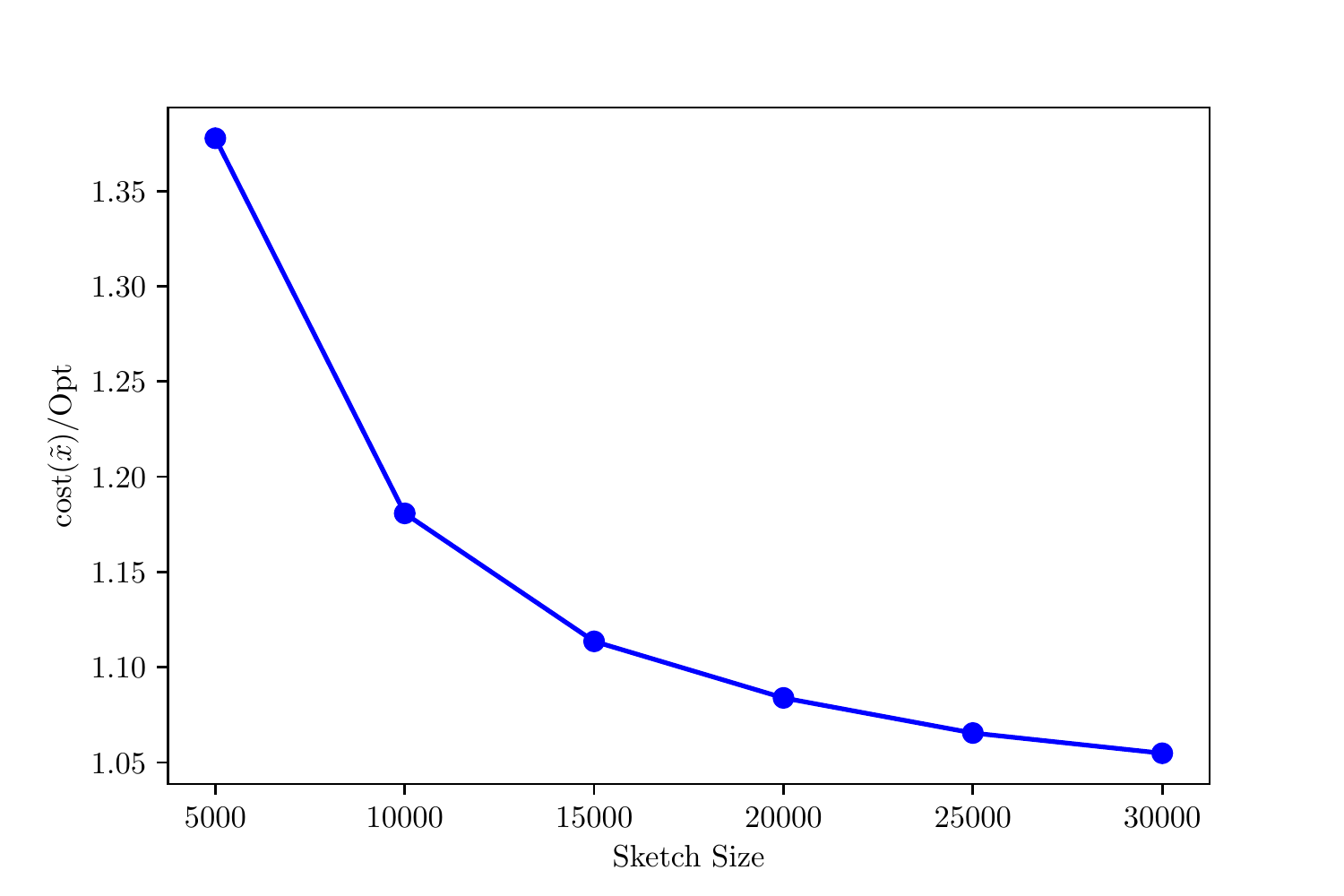}
  \caption{$\cost(\tilde{x})/\opt$ vs \# of rows of OSNAP}\label{fig:awesome_image2}
\endminipage\hfill\par
\minipage{0.3\textwidth}
\endminipage\hfill
\minipage{0.4\textwidth}%
  \includegraphics[width=\linewidth]{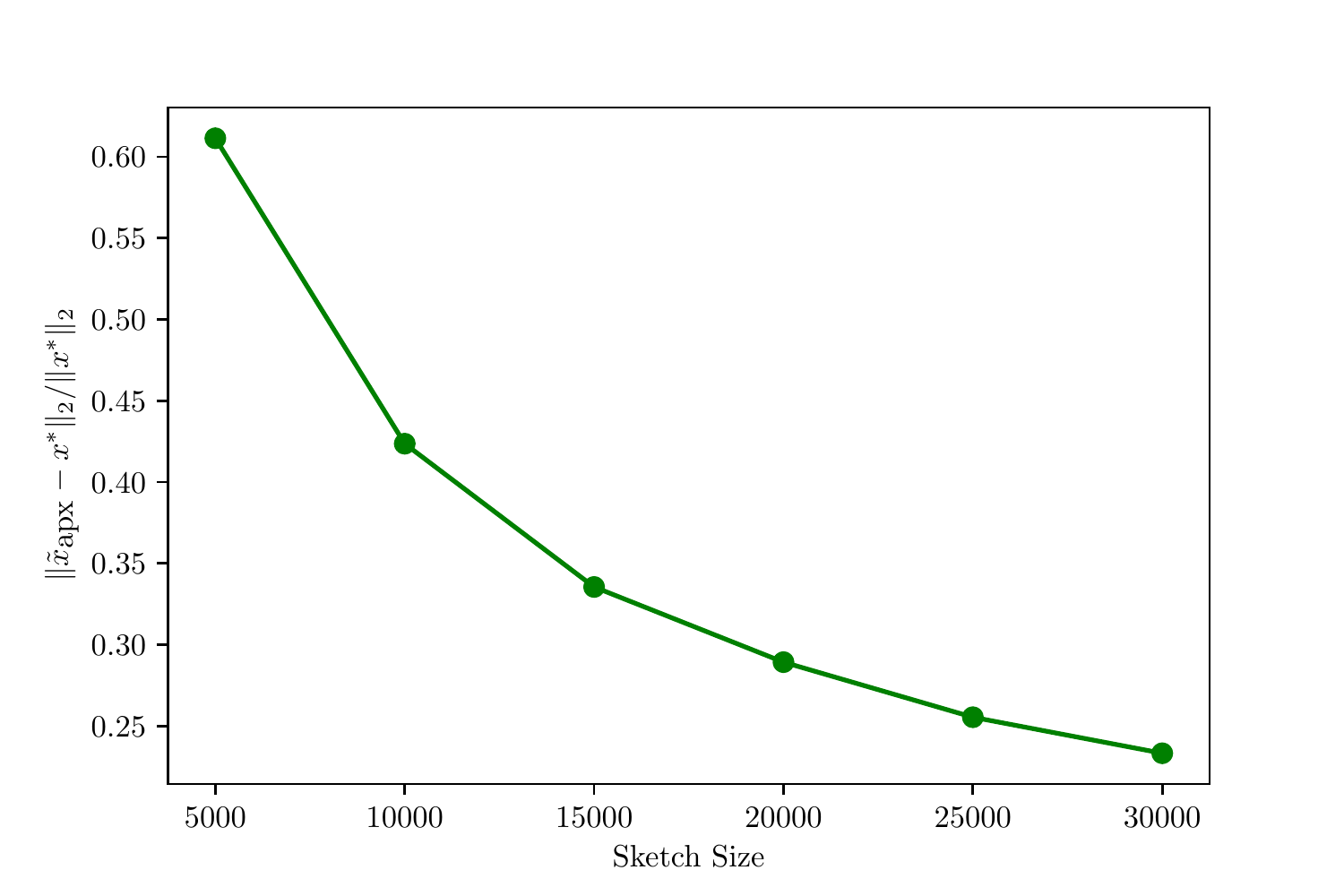}
  \caption{$\opnorm{\tilde{x} - x^*}/\opnorm{x^*}$ vs \# of rows of OSNAP}\label{fig:awesome_image3}
\endminipage\hfill
\minipage{0.3\textwidth}
\endminipage\hfill
\end{figure}